\documentclass[a4paper,11pt]{article}
\pdfoutput=1 

\usepackage{jheppub} 

\usepackage[T1]{fontenc} 
\usepackage{blkarray}
\usepackage{braket}
\usepackage{todonotes}
\usepackage{multirow}
\usepackage{simpler-wick}
\usepackage{amsthm}

\newtheorem{theorem}{Theorem}
\newtheorem{prop}{Proposition}

\title{\boldmath Flat Connections from Irregular Conformal Blocks}

\author[\ast,\dag]{Babak Haghighat}
\author[\ast]{Yihua Liu}
\author[\ast,\dag]{and Nicolai Reshetikhin}
\affiliation[\ast]{Yau Mathematical Sciences Center, Tsinghua University, Beijing, 100084, China}
\affiliation[\dag]{Yanqi Lake Beijing Institute of Mathematical Sciences and Applications (BIMSA), Huairou District, Beijing 101408, P. R. China}

\abstract{In this work we study Liouville conformal blocks with degenerate primaries and one operator in an irregular representation of the Virasoro algebra. Using an algebraic approach, we derive modified BPZ equations satisfied by such blocks and subsequently construct corresponding integral representations based on integration over non-compact Lefschetz cycles. The integral representations are then used to derive novel types of flat connections on the irregular conformal block bundle.}

\begin{document} 
\maketitle
\flushbottom

\section{Introduction and Summary}
Irregular conformal blocks in the two-dimensional Liouville CFT were first studied \cite{Gaiotto:2009ma,Gaiotto:2012sf,Bonelli:2011aa} in the context of the AGT correspondence \cite{Alday:2009aq}. Their importance arises from the fact that the quantum geometry of 4d $\mathcal{N}=2$ Argyres-Douglas theories \cite{Argyres:1995jj} admits higher order or \textit{irregular} singular points. To model such singularities from the viewpoint of the 2d CFT, one lets two conformal primaries collide in such a way that their OPE expansion with the energy momentum tensor gives rise to higher (than second) order poles. Using the operator state correspondence, such irregular operators can then be understood as eigenstates of the Virasoro generators $L_n$ for positive $n$. In other words they correspond to coherent states in the Hilbert sapce. 

There are two principal methods for the computation of conformal blocks with ordinary primary field insertions. One approach is to solve the so called \textit{BPZ equation} \cite{Belavin:1984vu} which is a higher order ordinary differential equation of Fuchsian type arising from the insertion of a degenerate state. Such degenerate states lead to null states upon action of a certain combination of Virasoro operators of a particular degree which upon using conformal Ward identities leads to the corresponding differential equation. Another method is the \textit{Coulomb gas} formalism of Dotsenko and Fateev \cite{Dotsenko:1984nm}. Here one views the individual primaries as charged particles which together with one or several screening charges are subject to a two-dimensional logarithmic Coulomb potential. The conformal block then arises from the computation of the corresponding partition function. The result can be expressed as an integral over certain loci in configuration space which mathematically can be viewed as period integrals. This approach is particularly useful for the interpretation of conformal blocks as fractional quantum Hall (FQH) wavefunctions \cite{Moore:1991ks,Santachiara:2010bt,VafaFQHE,Bergamin:2019dhg}. This gives a natural handle to compute the quantum statistics of Anyons in such systems and in particular to identify non-Abelian topological phases. These are gapped phases of matter where adiabatically transporting one Anyon around another gives rise to a unitary transformation of the degenerate groundstate wavefunctions. Such systems can therefore store topologically protected qubits and are important for topological quantum computation \cite{DasSarma:2005zz}. From a more modern perspective, there will be one-form symmetries \cite{Gaiotto:2014kfa} in the bulk acting on anyonic Wilson lines giving rise to non-invertible symmetries \cite{Shao:2023gho}.

The presence of irregular operators, with which we are concerned in the current paper, changes the above structures as follows. The conformal Ward identities are modified leading to modified versions of the BPZ equation \cite{Bonelli:2022ten}. We re-derive these modified BPZ equations by first constructing explicit representations of irregular vertex operators and subsequently compute their OPE with the energy momentum tensor. We then use the explicit form of the operators to compute the integrand of the Coulomb gas representation of irregular conformal blocks. The corresponding potential gets modified by polynomial terms in addition to the previously present logarithmic terms. From a more geometric perspective, conformal blocks of ordinary primaries correspond to Hypergeometric functions which can be viewed as sections of flat bundles over the configuration space. These sections have monodromies around regular singular points giving rise to representations of the braid group \cite{cmp/1104201923,Varchenko,Witten:1988hf}. This structure results in a braided fusion category and ultimately in a modular tensor category \cite{Moore-Seiberg}. Irregular vertex operators can be viewed as new types of monodromy defects and give rise to irregular singularities in configuration space. From a symmetry category point of view, these monodromy defects correspond to non-simple objects and are expected to give rise to new types of fusion categories \cite{Gu:2023bjg}. The resulting conformal block bundles are more interesting from several points of view. First, there will be Stokes lines emanating from the irregular singularity which change the monodromy behavior and result in asymptotic expansions there \cite{Gaiotto:2011nm,Gu:2023bjg}. Second, the corresponding flat connections have thus far never been studied and their integrability conditions have not been determined explicitly, yet. From yet another perspective, one expects a relation to the tt* equations of certain Landau-Ginzburg models with mixed logarithmic and polynomial potentials \cite{Cecotti:2014wea}. We explicitly construct flat connections by taking derivatives of the previously derived period integral representations of irregular conformal blocks.

\subsection{Coulomb gas formalism and Yang-Yang functional}
Our method to compute the conformal block is the free field realization also known as the Coulomb gas formalism \cite{Dotsenko:1984nm}. Considering a two-dimensional Liouville type CFT \cite{Teschner_2001,https://doi.org/10.48550/arxiv.hep-th/9304011}, with classical action given by
\begin{equation}
    S=\int dtd\sigma \left(\frac{1}{16\pi}((\partial_t\varphi)^2-(\partial_\sigma\varphi)^2)-\mu e^\varphi\right),
\end{equation}
the central charge and conformal dimensions of primary operators of the corresponding quantum theory are parameterized in terms of (a complex) coupling constant parameter $b$ as follows,
\begin{equation}
    Q = b + \frac{1}{b}, \quad c = 1 + 6 Q^2, \quad \Delta_{\alpha} = \alpha (Q - \alpha).
\end{equation}
Here, $\alpha$ is known as the momentum of the primary operator which is a vertex operator and can be expressed as $V_{\alpha} = e^{2\alpha \phi}$, with  $\phi\underset{b\rightarrow 0}{\sim}\frac{1}{2b}\varphi$ and $\alpha \in \mathbb{C}$. More useful for our purposes will be the following set of degenerate fields labeled by $r,s \in \mathbb{N}^+$, and constructed using the scalar field $\phi$ (see for example \cite{VafaFQHE}):
\begin{equation}
    \Psi_{r,s} \equiv \exp\left([(1-r)b + \frac{1-s}{b}]\phi\right).
\end{equation}
The corresponding conformal dimensions are given by
\begin{equation}
    h_{r,s} = h_0 - \frac{1}{4} \left(r b + s \frac{1}{b}\right)^2, \quad h_0 = \frac{1}{24}(c-1).
\end{equation}
For $k_a \in \mathbb{N}^+$ the relation $V_{-k_a/2b} = \Psi_{1,k_a+1}$.
These operators satisfy the fusion rules
\begin{equation}
    V_{-\frac{k_1}{2b}} \times V_{-\frac{k_2}{2b}} = \sum_{l = |k_1-k_2|}^{k_1 + k_2} V_{-\frac{l}{2b}}, \quad \Delta l = 2.
\end{equation}
Our goal is to compute conformal blocks $\mathcal{F}_{\Gamma_{\{\alpha\}}}$ arising in the computation of the correlation function of the following product of degenerate fields,
\begin{equation} \label{eq:pathint}
    \langle \prod_a V_{-k_a/2b}(z_a,\bar{z}_a)\rangle \equiv \int \left[\mathcal{D}\phi\right] e^{-S[\phi]} \prod_a e^{-\frac{k_a}{b} \phi}= \sum_{\{\alpha\}} C_{\{k\},\{\alpha\}} \overline{C}_{\{k\},\{\alpha\}}\mathcal{F}_{\Gamma_{\{\alpha\}}}\overline{\mathcal{F}}_{\Gamma_{\{\alpha\}}},
\end{equation}
where $\{\alpha\}$ denotes the collection of all fields arising as intermediate fusion states in the conformal block tree and  $C_{\{k\},\{\alpha\}}$ are suitable fusion coefficients. The conformal blocks $\mathcal{F}$ can be computed using extra insertions of $V_{1/b}(w_i)$ operators by integrating the following free field correlator over suitable integration cycles $\Gamma$ \cite{Dotsenko:1984nm,Teschner_2001,Gaiotto:2011nm},
\begin{equation} \label{eq:freefield}
    \begin{aligned} 
    \mathcal{V}\left(\prod_iV_{1/b}(w_i)\prod_aV_{-k_a/2b}(z_a)\right)&=\prod_{i<j}(w_i-w_j)^{\frac{-2}{b^2}}\prod_{i,a}(w_i-z_a)^{\frac{k_a}{b^2}}\prod_{a<b}(z_a-z_b)^{\frac{-k_ak_b}{2b^2}}\\
    &=\exp\left(\frac{1}{b^2}\mathcal{W}(w,z)\right), 
    \end{aligned} 
\end{equation}
where $\mathcal{W}$ is given by
\begin{equation} \label{eq:Wreg}
    \mathcal{W} = \sum_{i<j} -2 \log(w_i-w_j) + \sum_{i,a} k_a \log(w_i-z_a) - \sum_{a<b} \frac{1}{2} k_a k_b \log(z_a - z_b),
\end{equation}
and is known as the Yang-Yang superpotential due to its relation to 2d $\mathcal{N}=(2,2)$ Landau-Ginzburg models and integrable systems \cite{Jeong:2018qpc,Yang:1968rm}. The $V_{1/b}(w_i)$ are known as screening charges and their number is restricted by the following charge conservation condition 
\begin{equation}\label{eq:charge}
    \sum_i \alpha_i + \sum_a \alpha_a = Q, \quad \alpha_i = \frac{1}{b},
\end{equation}
which corresponds to poles in the evaluation of the path integral \eqref{eq:pathint} with residues given by the free field correlator \eqref{eq:freefield} in the vicinity of a background charge $-Q$. One may employ the duality $\alpha_{i,a} \mapsto Q - \alpha_{i,a}$, in order to ensure that the constraint \eqref{eq:charge} holds. One then integrates over $w_i$ variables along suitable cycles $\Gamma$ to arrive at an expression for the conformal blocks of the interacting Liouville CFT,
\begin{equation}
\label{block}
    \mathcal{F}_{\Gamma}  \equiv \int_\Gamma \mathcal{V}\left(\prod_iV_{1/b}(w_i)\prod_aV_{-k_a/2b}(z_a)\right) \prod_idw_i,
\end{equation}
where we may have more than one integration variable, denoted by $w_i$, and the integration cycle $\Gamma$ can be rather freely chosen as long as the integral converges. Let us give an example to clarify the procedure. To this end consider the 4-point function
\begin{eqnarray}
    \langle \Psi_{1,2}(\infty) \Psi_{1,2}(1) \Psi_{1,2}(z) \Psi_{1,2}(0)\rangle_{\Gamma} & = & \int_\Gamma \mathcal{V}\left(V_{\frac{1}{b}}(w) V_{-\frac{1}{2b}}(0)V_{-\frac{1}{2b}}(z)V_{-\frac{1}{2b}}(1)V_{-\frac{1}{2b}}(\infty)\right) dw \nonumber \\
    ~ & \sim & \left[z(1-z)\right]^{-\frac{1}{2b^2}} \int_{\Gamma} dw \left[w (w-1)(w-z)\right]^{1/b^2}, 
\end{eqnarray}
where $\Gamma$ can be one of two cycles $\Gamma_{\alpha}$, $\alpha=0,\frac{1}{b}$, defined as
\begin{equation}
    \Gamma_0 = \left[1,\infty\right], \quad \Gamma_{\frac{1}{b}} = \left[0,z\right].
\end{equation}
Thus the two integration cycles correspond to two conformal blocks which in turn correspond to the two different fusion channels. More generally, the integration cycles are given by \textit{Lefschetz thimbles} \cite{Gaiotto:2011nm,Gu:2023bjg}. Such integration thimbles are composed of steepest descent paths flowing out from critical points of $\mathcal{W}(w,z)$. So for a critical point $\sigma$ of $\mathcal{W}(w,z)$, the associated Lefschetz thimble is obtained by joining all paths $\mathcal{W}(t)$, $t \in (-\infty,0]$ which solve the first order equation
\begin{equation}
    \frac{d\bar{w}}{dt} = - \frac{\partial \mathcal{W}}{\partial w}, \quad \frac{dw}{dt} = - \frac{\partial \bar{\mathcal{W}}}{d \bar{w}}, \quad \mathcal{W}(-\infty) = \sigma. 
\end{equation}
Such thimbles can either arise from flows between two critical points but in more general situations are non-compact and flow to infinity along both ends. These non-compact cycles become important once we compute conformal blocks with irregular operators which is the main subject of the present paper.

\subsection{Summary of results}

The main results are as follows. In Section \ref{sec:irrvirasoro}, we revisit irregular representations of the Virasoro algebra and derive a modified BPZ equation for conformal blocks with three degenerate primaries $\Psi_{1,2}$ and one rank one irregular operator at infinity using an algebraic approach, see equation \eqref{eq:irrBPZ}. We then proceed in Section \ref{sec:intrep} to derive integral representations for conformal blocks with an arbitrary number of degenerate fields and screening charges and one irregular operator at infinity. The result is found to be given by
\begin{equation} \label{eq:Firr}
    \mathcal{F}_{\Gamma} = \int_{\Gamma} \exp\left(\frac{1}{b^2} \mathcal{W}_{\mathrm{irr}}\right)\prod_i dw_i,
\end{equation}
where $\mathcal{W}_{\mathrm{irr}}$ is given by
\begin{eqnarray}
    \mathcal{W}_{\mathrm{irr}} & \equiv & \sum_{i<j} -2 \log(w_i-w_j) + \sum_{i,a} k_a \log(w_i-z_a) - \sum_{a<b} \frac{1}{2} k_a k_b \log(z_a - z_b) \nonumber \\
    ~ & ~ & + \Lambda\left(\sum_i w_i - \frac{1}{2} \sum_a k_a z_a \right). \label{eq:Wirr}
\end{eqnarray}
In the above, $\Lambda$ is a constant proportional to the eigenvalue of the irregular state under the action of the Virasoro generator $L_1$. Furthermore, the integration cycle $\Gamma$ is a non-compact Lefschetz thimble corresponding to a critical point of $\mathcal{W}_{\mathrm{irr}}$. Restricting to three degenerate primaries of types $\Psi_{1,2}$, $\Psi_{1,k_0}$, $\Psi_{1,k_1}$ and one screening charge, we show that the above integral representation satisfies the modified BPZ equation 
\begin{align}
    &\left(b^2\frac{\partial^2}{\partial z^2}-(\frac{1}{z}+\frac{1}{z-1})\frac{\partial}{\partial z}+\frac{\Delta(k_1)}{(z-1)^2}+\frac{\Delta(k_0)}{z^2} \right. \nonumber\\
    &\left.+\frac{1}{z(z-1)}\left(\Lambda\frac{\partial}{\partial \Lambda}+\Delta_\alpha-\Delta(k_1)-\Delta(k_0)-\Delta(1)\right)-\frac{\Lambda(\alpha-Q)}{b z}-\frac{\Lambda^2}{4 b^2}\right)\mathcal{F}_\Gamma=0.
\end{align}
In Sections \ref{sec:diffeq} and \ref{sec:generaldiff} we then show that the integrals \eqref{eq:Firr} satisfy first order vector valued differential equations defining flat connections on the conformal block bundle. To state the result, we first rewrite $\mathcal{F}$ as 
\begin{equation}
    \mathcal{F}_{\Gamma} = \prod_{i}\exp(-\frac{\Lambda}{2 b^2}z_i)\prod_{i\neq j}(z_i-z_j)^{\frac{-k_i k_j}{2b^2}} \phi_{\Gamma}(z), 
\end{equation}
with
\begin{equation} \label{eq:phin-point}
    \phi_{\Gamma}(z) \equiv \int_{\Gamma} dw \prod_{i}((w-z_i)^{\frac{k_i}{b^2}})\exp\left(\frac{\Lambda}{ b^2}w\right). 
\end{equation}
Finally, defining $\Psi=\begin{pmatrix}
\frac{\partial \phi}{\partial z_1}\\ ...
\\\frac{\partial \phi }{\partial z_n}\\
\end{pmatrix}$, one finds that $\Psi$ satisfies the equation
\begin{align} 
\frac{\partial }{\partial z_i}\Psi=A_i \Psi +\sum_{j\neq i}\frac{\Omega_{ij}}{z_i-z_j}\Psi,
\end{align}
where $A_i$ and $\Omega_{ij}$ are given by equations \eqref{eq:A} and \eqref{eq:Omega}. It is shown in Section \ref{sec:generaldiff} that these matrices indeed satisfy the required flatness conditions.

\section{Irregular representations of the Virasoro algebra}
\label{sec:irrvirasoro}

Irregular vectors were initially introduced by Gaiotto \cite{Gaiotto:2009ma}, and are characterized as special states in the CFT Hilbert space that are eigenstates of Virasoro generators $L_n$ with positive $n$ with all but a finite set of eigenvalues vanishing. In terms of the Liouville CFT, the corresponding operators can be modeled as collision limits of two regular vertex operators as derived for example in \cite{Gaiotto:2012sf}. In the first part of this section we will be reviewing the definition of irregular states and, by using the operator state correspondence, provide corresponding vertex operator representations. In the second part, we derive, using an algebraic approach, modified BPZ equations satisfied by conformal blocks with 3 regular primaries and one irregular operator at infinity. 

\subsection{Definition}
Irregular states are classified by their rank which is the largest integer $r$ for which the eigenvalue of $L_{2r}$ is non-vanishing. Taking an algebraic point of view, in the simplest of rank $1$, the irregular state $\ket{I_{\alpha,c}}$ satisfies the following properties \cite{Gaiotto:2012sf}:
\begin{align} 
    L_0 \ket{I_{\alpha,c}}&=(\Delta_{\alpha^{}}+c\partial_c)\ket{I_{\alpha,c}}, \nonumber\\
   L_1\ket{I_{\alpha,c}}&=-c(\alpha^{}-Q)\ket{I_{\alpha,c}}, \nonumber \\
    L_2\ket{I_{\alpha,c}}&=-\frac{c^2}{4}\ket{I_{\alpha,c}},\label{eq:LI}
\end{align}
where $\alpha$ denotes the Liouville momentum and $c$ is an arbitrary complex parameter not to be confused with the central charge. In Liouville theory, as reviewed before, we can use the Coulomb gas formalism to compute the correlation functions. Thus the $n$-point function introduced before can be obtained through the OPE of the vertex operators alone. In the following, we make a proposal (based on earlier work in \cite{Gaiotto:2011nm}) for rank $1$ irregular vertex operators. The corresponding expression, using the free field $\chi$, is given by
\begin{align} \label{eq:irrop1}
    \exp(c \partial \chi(L)),\quad \text{with limit $L\rightarrow\infty$.}
\end{align}
To verify that the above vertex operator indeed gives rise to an irregular state, we compute its OPE with the energy momentum tensor which for Liouville theory is defined as follows,
\begin{align}
    T(z)=-:\partial\chi(z) \partial \chi(z):+  Q \partial^2\chi(z),\text{with \space} Q=b+\frac{1}{b}.
\end{align}
We can hence deduce the action of the Virasoro algebra on the irregular state by employing the state-operator correspondence and the OPE of the stress-energy tensor with the irregular operator \eqref{eq:irrop1}. Using 
 Wick's theorem and the correlation function of the free field:
 \begin{align}
 \braket{\chi(z)\chi(L)}=-\frac{1}{2}\log(z-L),
 \end{align}
 we thus obtain
\begin{align}
     T(z)\exp(c \partial \chi(L))&=(-:\partial\chi(z) \partial \chi(z):+Q \partial^2\chi(z)) \exp(c \partial \chi(L))\nonumber\\
    &=-(\langle\partial\chi(z)c\partial\chi(L)\rangle^2 \exp(c \partial \chi(L))\nonumber\\
    &-2\langle\partial\chi(z)c\partial\chi(L)\rangle\partial\chi(z)\exp(c \partial \chi(L)))+Q\partial^2\chi(z)\exp(c \partial \chi(L))\nonumber\\
    &=(-\frac{c^2}{4}\frac{1}{(z-L)^4}+c\frac{1}{(z-L)^2}\partial\chi(z))\exp(c \partial \chi(L))\nonumber\\
    &+Q\partial^2\chi(z)\exp(c \partial \chi(L)) \nonumber\\
    &=(-\frac{c^2}{4}\frac{1}{(z-L)^4}+c\frac{1}{(z-L)^2}\partial\chi(z))\exp(c \partial \chi(L))\nonumber\\
    &+Q c\frac{1}{(z-L)^3}\exp(c \partial \chi(L)) \nonumber\\
    &=\left(-\frac{c^2}{4}\frac{1}{(z-L)^4}+Q c\frac{1}{(z-L)^3}\right)\exp(c \partial \chi(L))\nonumber\\
    &+c\frac{1}{(z-L)^2}(\partial\chi(L)+(z-L)\partial^2\chi(L))\exp(c \partial \chi(L)).
\end{align}
To obtain the action of $L_0$ on the corresponding irregular vector, we just use the integral expression
\begin{align}
    L_n=\oint \frac{dz}{2\pi i} z^{n+1} T(z),
\end{align}
giving
\begin{align}
    [L_n,\exp(c \partial \chi(L))]=\oint_{\text{around\space} L} \frac{dz}{2\pi i} z^{n+1} T(z) \exp(c \partial \chi(L)).
\end{align}
For $n=0$, we have
\begin{align}
    [L_0,\exp(c \partial \chi(L))]&=\oint_{\text{around\space} L} \frac{dz}{2\pi i} z T(z) \exp(c \partial \chi(L))\nonumber\\
    &=\oint_{\text{around\space} 0} \frac{dz}{2\pi i} (z+L) (-\frac{c^2}{4}\frac{1}{z^4}+Q c\frac{1}{z^3}\nonumber\\
    &+c\frac{1}{z^2}(\partial\chi(L)+z\partial\chi(L)))\exp(c \partial \chi(L))\nonumber\\
    &=c\partial\chi(L)\exp(c \partial \chi(L))+c L\partial^2\chi(L)\exp(c \partial \chi(L))\nonumber\\
    &=c \partial_c \exp(c \partial \chi(L))+c L \partial \exp(c \partial \chi(L)).
\end{align}
Similarly, for $ n=1$, we obtain
\begin{align}
    [L_1,\exp(c \partial \chi(L))]&=(Qc\exp(c \partial \chi(L)))+2c L \partial_c \exp(c \partial \chi(L))+c L^2 \partial \exp(c \partial \chi(L)).
\end{align}
From the above equations, we can get the action of the conformal operators via the state-operator correspondence by letting $L\rightarrow 0$, that is
\begin{align}
   [L_0,\exp(c \partial \chi(0))]\ket{0}&=c \partial_c \ket{I_{c}},\\
    [L_1,\exp(c \partial \chi(0))]\ket{0}&=Qc\ket{I_{c}},\\
    [L_2,\exp(c \partial \chi(0))]\ket{0}&=-\frac{c^2}{4}\ket{I_{c}},
\end{align}
where
\begin{equation}
    \ket{I_{c}} \equiv \exp(c \partial \chi(0)) \ket{0}.
\end{equation}
We can then see that this irregular state corresponds to the one from the definition given in \eqref{eq:LI} for $\alpha=0$. On the other hand, we can also construct vertex operators that satisfy the definition with nonzero $\alpha$, by setting
\begin{align}
    V(L) \equiv ~:\exp(2\alpha \chi(L))\exp(c \partial \chi(L)): .
\end{align}
We can check that the OPE with the stress-energy tensor has the following form
\begin{align}
    &T(z) :\exp(2\alpha \chi(L))\exp(c \partial \chi(L)): \nonumber \\
    &=(-:\partial\chi(z) \partial \chi(z):+Q \partial^2\chi(z)):\exp(2\alpha \chi(L))\exp(c \partial \chi(L)):\nonumber\\
    &=\left(\frac{Q\alpha-\alpha^2}{(z-L)^2}+\frac{\alpha}{z-L}\partial\chi(z)\right)V(L)-\frac{c\alpha}{(z-L)^3}V(L) \nonumber\\  
&+\left(-\frac{c^2}{4}\frac{1}{(z-L)^4}+c\frac{1}{(z-L)^2}\partial\chi(z)\right)V(L)+Q c\frac{1}{(z-L)^3}V(L)  .  
\end{align}
Similar to before, we can then read off the action of $L_i$ on the irregular state using the state-operator correspondence,
\begin{align}
    L_0\ket{I_{\alpha,c}}&=\Delta_\alpha+c\partial_c\ket{I_{\alpha,c}}, \\
    L_1\ket{I_{\alpha,c}}&=-c(\alpha-Q)\ket{I_{\alpha,c}}, \\
    L_2\ket{I_{\alpha,c}}&=-\frac{c^2}{4}\ket{I_{\alpha,c}},
\end{align}
where $\ket{I_{\alpha,c}}=V(0)\ket{0}$, which is exactly the irregular vector of rank $1$ introduced in \cite{Gaiotto:2012sf}.

\subsection{Derivation of modified BPZ equation}
We can derive BPZ like equations utilizing the above OPEs. When considering the descendant states of the primary state corresponding to  $V_{\frac{-1}{2b}}(0)\ket{0}:=  \ket{V_{\frac{-1}{2b}}}$ with conformal weight $\Delta=-\frac{1}{2b}(Q+\frac{1}{2b})$, one finds that there is a null state, which means a state of zero norm,  of the following form
\begin{align}
    (L_{-2}+b^2L_{-1}^2)\ket{V_{\frac{-1}{2b}}}.
\end{align}
Inserting the above null state into the correlation function should thus lead to vanishing of the correlation function,
\begin{align}
\mathcal{F}=\bra{0}V(\infty)(L_{-2}+b^2L_{-1})V_{\frac{-1}{2b}}(z)V_{\frac{-1}{2b}}(z_0)V_{\frac{-1}{2b}}(z_1)\ket{0}=0,
\end{align}
which corresponds to the following equation
\begin{align}
\mathcal{L}_{-2}\mathcal{F}+b^2\frac{\partial^2}{\partial z^2}\mathcal{F}=0.
\end{align}
The operator $\mathcal{L}_{-2}$ is called the descent operator that directly creates the descendant state by employing the state-operator correspondence, and whose action on a primary operator $\phi$ is defined as follows
\begin{align}
    \mathcal{L}_{-2} \phi(z)=\oint \frac{dw}{2\pi i}\frac{1}{(w-z)}T(w)\phi(z).
\end{align}
In the case when all fields in the bracket are primary fields, the operator $\mathcal{L}_{-2}$ has the representation
\begin{align}
    \mathcal{L}_{-2}=\left ( \sum_i \frac{h_i}{(z_i-z)^2}-\frac{1}{z_i-z}\frac{\partial}{\partial z_i} \right ).
\end{align}
This can be derived by considering the definition of the descendant fields in terms of contour integrals of the primary field OPE with the stress energy tensor $T(w)$ where the contour in the $w$-plane goes around the singularities at $z_i$. Since in our case, there is an additional irregular vertex operator, the  representation of the action of $\mathcal{L}_{-2}$ should be modified. Computing directly, one finds the action on the irregular operator in the following way
\begin{align}
    &\mathcal{L}_{-2}V(\infty)=-\oint_{\text{around $\infty$}} \frac{dw}{2\pi i }\frac{1}{w-z} T(w)  V(\infty) \nonumber\\
    &=-\oint_{\text{around $\infty$}} \frac{dw}{2\pi i }\frac{1}{w} T(w)  V(\infty) \nonumber\\
    &\stackrel{v=\frac{1}{w}}{=}-\oint_{\text{around $0$}}\frac{dv}{2\pi i }\frac{dw}{dv}  v \left(\frac{dv}{dw}\right)^2 T(v)  V \nonumber\\
    &=\oint\frac{dv}{2\pi i }  v^3\left(\left(\frac{Q\alpha-\alpha^2}{v^2}+\frac{\alpha}{v}\partial\chi(0)\right)V-\frac{c\alpha}{v^3}V +(-\frac{c^2}{4}\frac{1}{v^4}+c\frac{1}{v^2}\partial\chi(0))V+Q c\frac{1}{v^3}V\right) \nonumber\\
    &=-\frac{c^2}{4}V.
\end{align}
Having obtained the action of the operator $\mathcal{L}_{-2}$ on the irregular state, we can proceed to deduce its action on the entire correlator and hence arrive at the second order differential equation that the conformal block $\mathcal{F}$ satisfies
\begin{align}
    \left((b^2\frac{\partial^2}{\partial z^2}+\frac{\Delta}{(z_1-z)^2}+\frac{\Delta}{(z_0-z)^2}-\frac{1}{z_1-z}\frac{\partial}{\partial z_1}-\frac{1}{z_0-z}\frac{\partial}{\partial z_0})-\frac{c^2}{4}\right)\mathcal{F}=0. \label{eq:irrBPZ1}
\end{align}
We can derive further equations satisfied by our conformal blocks corresponding to global conformal Ward identities. Notice that global conformal invariance with respect to $L_{-1}$ implies the identity
\begin{align}
    \mathcal{L}_{-1} \mathcal{F}(z,z_0,z_1,c)
    &=\bra{0}[L_1,I_{\alpha,c}]V_{-\frac{1}{2b}}(z)V_{-\frac{1}{2b}}(z_0)V_{-\frac{1}{2b}}(z_1)\ket{0}\nonumber\\&+\bra{0}I_{\alpha,c}[L_{-1},V_{-\frac{1}{2b}}(z)]V_{-\frac{1}{2b}}(z_0)V_{-\frac{1}{2b}}(z_1)\ket{0}
    \nonumber\\&+\bra{0}I_{\alpha,c} V_{-\frac{1}{2b}}(z)[L_{-1},V_{-\frac{1}{2b}}(z_0)]V_{-\frac{1}{2b}}(z_1)\ket{0}\nonumber\\&+\bra{0}I_{\alpha,c} V_{-\frac{1}{2b}}(z)V_{-\frac{1}{2b}}(z_0)[L_{-1},V_{-\frac{1}{2b}}(z_1)]\ket{0}=0.
\end{align}
Here the irregular state $\bra{I_{\alpha,c}}$ is inserted at infinity, so the action of $L_{-1}$ becomes its adjoint. From the action above, we have
\begin{align}
    L_1 \ket{I_{\alpha,c}}=-c(\alpha-Q) \ket{I_{\alpha,c}}=(b+\frac{1}{b}-\alpha)c\ket{I_{\alpha,c}}.
\end{align}
Another global conformal Ward identity is obtained from the action of $L_0$. With these two equations, we can express the partial derivatives $\frac{\partial}{\partial z_1}$ and $\frac{\partial}{\partial z_0}$ in terms of differential operators with respect to $z$ and $c$, i.e.
\begin{align}
    &\text{$L_{-1} $ invariance}\Rightarrow \left(\frac{\partial}{\partial z}+\frac{\partial}{\partial z_0}+\frac{\partial}{\partial z_1}+c(\alpha-Q) \right)\mathcal{F}=0. \label{eq:L1inv} \\
    & \text{$L_0 $ invariance}\Rightarrow \left(z\frac{\partial}{\partial z}+z_0\frac{\partial}{\partial z_0}+z_1\frac{\partial}{\partial z_1}+3\Delta -(c\frac{\partial}{\partial c}+\Delta_\alpha)\right)\mathcal{F}=0. \label{eq:L0inv}
\end{align}
Thus giving
\begin{align}
    \frac{\partial}{\partial z_1}\mathcal{F}&= \left(\frac{\Delta _{\alpha }-3 \Delta}{(z_1-z_0)}+z_0\frac{c(\alpha-Q)}{(z_1-z_0)}+\frac{c}{(z_1-z_0)}\frac{\partial}{\partial c}-\frac{z-z_0}{z_1-z_0}\frac{\partial}{\partial z}\right) \mathcal{F}, \\
     \frac{\partial}{\partial z_0}\mathcal{F}&=\left(\frac{\Delta _{\alpha }-3 \Delta}{(z_0-z_1)}+z_1\frac{c(\alpha-Q)}{(z_0-z_1)}+\frac{c}{(z_0-z_1)}\frac{\partial}{\partial c}-\frac{z-z_1}{z_0-z_1}\frac{\partial}{\partial z}\right)\mathcal{F}.
\end{align}
Inserting back in \eqref{eq:irrBPZ1} finally gives the BPZ like differential equation satisfied by $\mathcal{F}$, namely
\begin{align}
    &\left(b^2\frac{\partial^2}{\partial z^2}+\frac{\Delta}{(z_1-z)^2}+\frac{\Delta}{(z_0-z)^2}-\frac{c^2}{4}\right)\mathcal{F}\nonumber\\
    &+\left(\frac{-2 z+z_0+z_1}{2 \left(z-z_0\right) \left(z-z_1\right)}\frac{\partial}{\partial z}+\frac{1}{ \left(z-z_0\right) \left(z-z_1\right)}c\frac{\partial}{\partial c} +\frac{-z+z_0+z_1}{ \left(z-z_0\right) \left(z-z_1\right)}c(\alpha-Q) \right.\nonumber\\
    &\left.+\frac{1}{ \left(z-z_0\right) \left(z-z_1\right)}(\Delta _{\alpha }-3 \Delta)\right)\mathcal{F}=0.
\end{align}
Setting $z_1=1,z_0=0$, we arrive at the final form of the BPZ equation:
\begin{align} \label{eq:irrBPZ}
    &\left(b^2\frac{\partial^2}{\partial z^2}-(\frac{1}{z}+\frac{1}{z-1})\frac{\partial}{\partial z}+\frac{\Delta}{(z-1)^2}+\frac{\Delta}{z^2}+ \frac{1}{z(z-1)}(c\frac{\partial}{\partial c}+\Delta_\alpha-3\Delta)-\frac{c(\alpha-Q)}{z}-\frac{c^2}{4}\right)\mathcal{F}=0.
\end{align}


\section{Integral representations}
\label{sec:intrep}

We next want to compute integral representations of our irregular conformal blocks. Our strategy will be to compute the free field correlator with one irregular operator insertion at infinity, 
\begin{equation}
    \mathcal{V}\left(\prod_i V_{1/b}(w_i) \prod_a V_{-k_a/2b}(z_a) I_{\alpha,c}(\infty) \right),
\end{equation}
and then subsequently write down a Coulomb gas representation using this correlator. To this end, we need to consider both the Wick contraction of an irregular operator at position $L$ with a regular vertex operator at position $z$ as well as the contraction of a regular vertex operator at $L$ with that at $z$. But the second contraction will only bring an overall factor depending on $L$ in front of the correlation function, when $L\rightarrow \infty$. So we only need to consider the first contraction, which is of the form
\begin{equation}
    \wick{\c1 \exp(2 a \chi(z)) \c1 \exp(c \partial \chi(L))} = \wick{\exp\left(2 a c \c1 \chi(z) \partial \c1 \chi(L)\right)} = \exp\left(-a c \frac{1}{z - L}\right)~.
\end{equation}
If we now make the following replacement,
\begin{equation}\label{coordtrans}
    c \mapsto \frac{L^2}{b} \Lambda,
\end{equation}
then the above can be simplified as follows for large $L$,
\begin{equation}
    \wick{\c1 \exp(2a \chi(z)) \c1 \exp(c \partial \chi(L))}  = \exp\left(a \frac{\Lambda L}{b}(1+   \frac{z}{L} + \mathcal{O}(L^{-2}))\right)\nonumber\propto \exp\left (\frac{a \Lambda}{b}   z  \right ).
\end{equation}
With this result and the OPE between regular vertex operators, it is now clear that the free field correlation functions in the integrand take the form
\begin{align}\label{Correlationfunction}
  &\mathcal{V}\left(\prod_i V_{1/b}(w_i) \prod_a V_{-k_a/2b}(z_a) I_{\alpha,c}(\infty) \right)\nonumber\\
  &\propto \prod_{i\neq j} (w_i-w_j)^{-\frac{2}{b^2}}\prod_{i,a} (w_i-z_a)^{\frac{k_a}{b^2}}\prod_{a\neq b} (z_a-z_b)^{\frac{-k_a k_b}{2b^2}}\exp\left(\frac{\Lambda}{b^2}(\sum_i w_i -\frac{1}{2}(\sum_a k_a z_a))\right).
\end{align}
The full integral representation then becomes
\begin{eqnarray} 
    \mathcal{F}_{\Gamma} & = &\int_{\Gamma}\mathcal{V}\left(\prod_i V_{1/b}(w_i) \prod_a V_{-k_a/2b}(z_a) I_{\alpha,c}(\infty) \right) \prod_i dw_i \nonumber \\
    ~ & = & \int_{\Gamma} \exp\left(\frac{1}{b^2} \mathcal{W}_{\mathrm{irr}}\right) \prod_i dw_i,
\end{eqnarray}
where $\mathcal{W}_{\mathrm{irr}}$ is given by \eqref{eq:Wirr}.
In the following sections we will be focusing on cases where there is only one screening charge, i.e. the integration cycles $\Gamma$ is a real co-dimension one cycle in the complex $w$-plane and show that the corresponding expression satisfies the modified BPZ equation.

\subsection{3-point function with irregular operator at infinity}

Let us next turn to the case of a 3-point function with one irregular operator at infinity.  Setting the $k_i = 1$, this case has superpotential
\begin{align}
 \mathcal{W}(w,z_1,z_2,z_3)=&-\frac{1}{2}(\log(z_1-z_2)+\log(z_1-z_3)+\log(z_2-z_3))+\\ \nonumber &\log(w-z_1)+\log(w-z_2)+\log(w-z_3)- \Lambda \\ \nonumber
 &\left(\frac{z_1}{2}+\frac{z_2}{2}+\frac{z_3}{2}-w\right),
\end{align}
and the conformal blocks can be written as the integral over the corresponding Lefschetz thimbles. For example, in this case, one can use conformal transformations setting $z_1=z,z_2=1,z_3=0$. Then our superpotential becomes
$$\mathcal{W}(w,z)=\log (w-z)+\log (w-1)+\log (w)+\Lambda \left(w-\frac{z}{2}-\frac{1}{2}-\frac{0}{2}\right)-\frac{1}{2}\log (z (z-1)),$$
and the resulting conformal block has the following thimble representation,
\begin{align}
    \mathcal{F}(z)&=\int_{\Gamma} dw \exp(\frac{1}{b^2} \mathcal{W}(w,z))\nonumber\\
    &=[z (z-1)]^{-\frac{1}{2 b^2}} e^{-\frac{\Lambda (z+1)}{2 b^2}}\int_{\Gamma} dw \exp \left(\frac{\Lambda w+\log (w-z)+\log (w-1)+\log (w)}{b^2}\right).
\end{align} 

\subsection{BPZ equation}

Now let us show that the integral representation of our conformal blocks satisfies the BPZ equation \eqref{eq:irrBPZ}. To this end, we first work with the generalized expression where the primaries are located at points $z$, $z_0$ and $z_1$. Then the corresponding block is given by
\begin{align}
    \mathcal{F}(z)&=\Lambda^{\beta}[(z-z_0) (z-z_1)(z_1-z_0)]^{-\frac{1}{2 b^2}} e^{-\frac{\Lambda (z+z_1+z_0)}{2 b^2}}
    \nonumber\\
    &\times \int_{\Gamma} dw \exp \left(\frac{\Lambda w+\log (w-z)+\log (w-z_1)+\log (w-z_0)}{b^2}\right),\nonumber
\end{align}
where we have introduced an overall scaling factor $\Lambda^{\beta}$ whose significance will become clear later. We want to show that $\mathcal{F}$ satisfies \eqref{eq:irrBPZ1}. To this end, we introduce the notation 
\begin{equation} 
    \mathcal{F}(z) = f(z,z_0,z_1,\Lambda)\times \phi(z,z_0,z_1,\Lambda),
\end{equation}
where 
\begin{align}
    \phi(z,z_0,z_1,\Lambda)&=\int_{\Gamma} dw \exp \left(\frac{\Lambda w+\log (w-z)+\log (w-z_1)+\log (w-z_0)}{b^2}\right)\\& \equiv \int_{\Gamma}dw \exp(I), \label{eq:phi3point}
\end{align}
and
\begin{equation} \label{eq:f3point}
    f(z,z_0,z_1,\Lambda) \equiv \Lambda^{\beta}[(z-z_0) (z-z_1)(z_1-z_0)]^{-\frac{1}{2 b^2}} e^{-\frac{\Lambda (z+z_1+z_0)}{2 b^2}}.
\end{equation}
Henceforth we will abbreviate $\phi(z,1,0,\Lambda) = \phi(z)$. Then equation \eqref{eq:irrBPZ1} becomes
    \begin{align}
    \left((b^2\frac{\partial^2}{\partial z^2}+\frac{\Delta}{(z_1-z)^2}+\frac{\Delta}{(z_0-z)^2}-\frac{1}{z_1-z}\frac{\partial}{\partial z_1}-\frac{1}{z_0-z}\frac{\partial}{\partial z_0})-\frac{\Lambda^2}{4b^2}\right)(f\phi)=0, \label{eq:irrBPZ2}
\end{align}
since  we performed the rescaling \eqref{coordtrans}, 
so now the value of $c$ under the transformation should be  $\frac{\Lambda}{b}$. And we notice by direct substitution that
\begin{align}
     \left((b^2\frac{\partial^2}{\partial z^2}+\frac{\Delta}{(z_1-z)^2}+\frac{\Delta}{(z_0-z)^2}-\frac{1}{z_1-z}\frac{\partial}{\partial z_1}-\frac{1}{z_0-z}\frac{\partial}{\partial z_0})-\frac{\Lambda^2}{4b^2}\right)f=0.
\end{align}
Thus it remains to prove the following identity
\begin{align}
    f\left(b^2\frac{\partial^2}{\partial z^2}-\frac{1}{z_1-z}\frac{\partial}{\partial z_1}-\frac{1}{z_0-z}\frac{\partial}{\partial z_0}\right)\phi+ 2b^2 \frac{\partial f}{\partial z}\frac{\partial \phi}{\partial z}=0.
\end{align}
As we will see in the next subsection and in more generality in Section \ref{sec:generaldiff}, the functions $\phi$ are very special and satisfy certain identities which bear resemblance to KZ-type equations defining flat connections on the conformal block bundle. In particular, identity \eqref{eq:phidoubleprime} gives
\begin{align} \label{eq:phiddz}
   b^2\frac{\partial^2 \phi}{\partial z^2}=\Lambda \frac{\partial \phi}{\partial z}+\frac{1}{z-z_0}(\frac{\partial \phi}{\partial z}-\frac{\partial \phi}{\partial z_0})+\frac{1}{z-z_1}(\frac{\partial \phi}{\partial z}-\frac{\partial \phi}{\partial z_1}),
\end{align}
or after rearranging,
\begin{align}
    \left(b^2\frac{\partial^2}{\partial z^2}-\frac{1}{z_1-z}\frac{\partial}{\partial z_1}-\frac{1}{z_0-z}\frac{\partial}{\partial z_0}\right)\phi=\left(\Lambda+\frac{1}{z-z_0}+\frac{1}{z-z_1}\right)\frac{\partial \phi}{\partial z}.
\end{align}
We also notice that
\begin{align}
    \frac{\partial f}{\partial z}=-\frac{1}{2b^2}(\Lambda+\frac{1}{z-z_0}+\frac{1}{z-z_1})f,
\end{align}
and therefore
\begin{align}
    &f\left(b^2\frac{\partial^2}{\partial z^2}-\frac{1}{z_1-z}\frac{\partial}{\partial z_1}-\frac{1}{z_0-z}\frac{\partial}{\partial z_0}\right)\phi+ 2b^2 \frac{\partial f}{\partial z}\frac{\partial \phi}{\partial z}\nonumber\\
    &=f\left(\Lambda+\frac{1}{z-z_0}+\frac{1}{z-z_1}\right)\frac{\partial \phi}{\partial z}+2b^2 \frac{\partial f}{\partial z}\frac{\partial \phi}{\partial z}\nonumber\\
    &=f\left(\Lambda+\frac{1}{z-z_0}+\frac{1}{z-z_1}\right)\frac{\partial \phi}{\partial z}+2b^2 (-\frac{1}{2b^2})(\Lambda+\frac{1}{z-z_0}+\frac{1}{z-z_1})f \frac{\partial \phi}{\partial z}\nonumber\\
    &=0.
\end{align}
This completes the proof of \eqref{eq:irrBPZ2}. Furthermore, from the global conformal invariance equations \eqref{eq:genirrBPZ2} and \eqref{eq:genirrBPZ3} derived in Appendix \ref{sec:globalinv}, we have that 
\begin{align}
     \left(\frac{\partial}{\partial z}+\frac{\partial}{\partial z_0}+\frac{\partial}{\partial z_1}+\frac{\Lambda}{b}(\alpha-Q) \right)\mathcal{F}=0,
\end{align}
and
\begin{align}
    \left(z\frac{\partial}{\partial z}+z_0\frac{\partial}{\partial z_0}+z_1\frac{\partial}{\partial z_1}+3\Delta -(\Lambda\frac{\partial}{\partial \Lambda}+\Delta_\alpha)\right)\mathcal{F}=0,
\end{align}
and the previously introduced constant $\beta$ is fixed to be $\beta = -1$ and $\alpha=b+\frac{3}{2b}$ (see Appendix \ref{sec:globalinv}). The final form of the BPZ equation is now obtained by solving for $\partial_{z_i} \mathcal{F}$, $i=0,1$ using the equations above and inserting back into \eqref{eq:irrBPZ1}. The result is
\begin{align}
    &\left(b^2\frac{\partial^2}{\partial z^2}+\frac{\Delta}{(z_1-z)^2}+\frac{\Delta}{(z_0-z)^2}-\frac{\Lambda^2}{4b^2}\right)\mathcal{F}\nonumber\\
    &+\left(\frac{-2 z+z_0+z_1}{2 \left(z-z_0\right) \left(z-z_1\right)}\frac{\partial}{\partial z}+\frac{1}{ \left(z-z_0\right) \left(z-z_1\right)}\Lambda\frac{\partial}{\partial \Lambda} +\frac{-z+z_0+z_1}{ \left(z-z_0\right) \left(z-z_1\right)}(\alpha-Q)\frac{\Lambda}{b} \right.\nonumber\\
    &\left.+\frac{1}{ \left(z-z_0\right) \left(z-z_1\right)}(\Delta _{\alpha }-3 \Delta)\right)\mathcal{F}=0.
\end{align}
Finally, setting $z_1=1,z_0=0$, we have

\begin{align}
\left(b^2\frac{\partial^2}{\partial z^2}-(\frac{1}{z}+\frac{1}{z-1})\frac{\partial}{\partial z}+\frac{\Delta}{(z-1)^2}+\frac{\Delta}{z^2}+\frac{1}{z(z-1)}(\Lambda\frac{\partial}{\partial \Lambda}+\Delta_\alpha-3\Delta)-\frac{\Lambda(\alpha-Q)}{b z}-\frac{\Lambda^2}{4b^2}\right)\mathcal{F}=0
\end{align}

as a special case of \eqref{eq:genirrBPZ}.

\section{Derivation of KZ-type equation}
\label{sec:diffeq}

This section is concerned with the derivation of KZ-type differential equations satisfied by irregular conformal blocks. We derive vector valued first order equations which define flat connections on the irregular conformal block bundle. We first derive such equations for 2-point functions with one additional irregular operator, 3-point functions with one irregular operator, and then subsequently generalize to the n-point case. 

As we have seen, when including an irregular operator at infinity, the expression in \eqref{eq:Wreg} for $\mathcal{W}$ gets modified in the following way,
\begin{equation} \label{eq:irrpot}
    \mathcal{W}(w,z)\rightarrow \mathcal{W}(w,z)+\Lambda(\sum_i w_i-\frac12\sum_a z_a),
\end{equation}
where $\Lambda$ is a constant proportional to the eigenvalue of $L_1$ when acting on the irregular state. The added term is called the symmetry breaking term \cite{Gaiotto:2011nm}. In the following, we will see how it gives rise to an irregular singularity at infinity in the monodromy representation of conformal blocks.

\subsection{2-point function with one irregular operator}

As a warm up, let us first have a look at the example of 2-point functions with one additional irregular operator. The superpotential is given as follows, corresponding to the insertion of two regular operators at $z_1$ and $z_2$ and an irregular operator at infinity,
\begin{align}
    \mathcal{W}(w,z_1,z_2)=-\frac{1}{2}(\log(z_1-z_2))+\log(w-z_1)+\log(w-z_2)-\Lambda(\frac{z_1}{2}+\frac{z_2}{2}-w).
\end{align}
In this case, all the conformal blocks corresponding to such a superpotential are given by integration over Lefschetz Thimbles \cite{Gaiotto:2011nm}:
\begin{align}
\phi(z_1,z_2)=\int_{\Gamma_\alpha} \exp{\frac{1}{b^2}W} dw,
\end{align}
where the Lefschetz thimbles $\Gamma_{\alpha}$ are determined by solving the gradient descent equations corresponding to BPS solitons from the Landau-Ginzburg point of view. Moreover, by conformal symmetry,  we can take $z_1=z,z_2=0$ utilizing a conformal transformation. So the integration can be rewritten as follows,
\begin{align}
\phi(z)=\int_{\Gamma_\alpha} \exp{\frac{1}{b^2}W} dw=z^\frac{1}{2b^2}\int_{\Gamma_\alpha}(w-z)^\frac{1}{b^2}w^\frac{1}{b^2} \exp{\frac{\Lambda}{b^2}(-\frac{z}{2}+w)} dw.
\end{align}
Performing the suitable substitutions $w=\frac{z}{2}(t+1)\Leftrightarrow w-\frac{z}{2}=\frac{z}{2}t$, the integral becomes
 \begin{align}
    \phi(z)=z^{\frac{5}{2b^2}}\int_{\Gamma_\alpha}(t^2-1)^\frac{1}{b^2} \exp{\frac{\Lambda zt}{2b^2}} dt.
 \end{align}
We now have the following
\begin{prop}
Let $\phi(z)$ be defined as above. Furthermore, define 
\begin{align}
    \phi_1(z)\equiv \int_{\Gamma_\alpha}t(t^2-1)^\frac{1}{b^2} \exp{\frac{\Lambda zt}{2b^2}} dt.
\end{align}
Then the vector $\varphi\equiv \begin{pmatrix}
\phi(z)\\ \phi_1(z)
\end{pmatrix}$, satisfies the following first order differential equation,
$$\frac{d\varphi}{dz}=\left(
\begin{array}{cc}
 \frac{\frac{5}{2 b^2}}{z} & \frac{\Lambda }{2 b^2} \\
 \frac{\Lambda }{2 b^2} & \frac{\frac{5}{2 b^2}}{z}-\frac{2 \left(\frac{1}{b^2}+1\right)}{z} \\
\end{array}
\right)\varphi=\left(\left(
\begin{array}{cc}
 0 & \frac{\Lambda }{2 b^2} \\
 \frac{\Lambda }{2 b^2} & 0 \\
\end{array}
\right)+\frac{1}{z}\left(
\begin{array}{cc}
 \frac{5}{2 b^2} & 0 \\
 0 & \frac{1}{2 b^2}-2 \\
\end{array}
\right)\right)\varphi~.$$
\end{prop}

\begin{proof}
We observe  that,
$$\phi^{'}(z)=\frac{\frac{5}{2b^2}}{z} \phi(z)+z^{\frac{5}{2b^2}}\int_{\Gamma_\alpha}(t^2-1)^\frac{1}{b^2} \exp{\frac{\Lambda zt}{2b^2}} \frac{\Lambda t}{2b^2} dt=\frac{\frac{5}{2b^2}}{z} \phi(z)+\frac{\Lambda}{2b^2}\phi_1(z),$$
Moreover, differentiating further $\phi_1$, one arrives at
\begin{align*}
    \phi_1^{'}(z)&=\frac{\frac{5}{2b^2}}{z} \phi_1(z)+z^{\frac{5}{2b^2}}\int_{\Gamma_\alpha}(t^2-1)^\frac{1}{b^2}\frac{\Lambda t^2}{2b^2} \exp{\frac{\Lambda zt}{2b^2}}  dt \\&=\frac{\frac{5}{2b^2}}{z} \phi_1(z)+\Lambda\frac{z^{\frac{5}{2b^2}}}{2b^2}\int_{\Gamma_\alpha}(t^2-1)^\frac{1}{b^2}(t^2-1)\exp{\frac{\Lambda zt}{2b^2}}dt+\Lambda\frac{z^{\frac{5}{2b^2}}}{2b^2}\int_{\Gamma_\alpha}(t^2-1)^\frac{1}{b^2}\exp{\frac{\Lambda zt}{2b^2}}dt\\&=\frac{\frac{5}{2b^2}}{z}\phi_1(z)+\Lambda\frac{z^{\frac{5}{2b^2}}}{2b^2}\int_{\Gamma_\alpha}(t^2-1)^{\frac{1}{b^2}+1}\frac{2b^2}{\Lambda z}d(\exp{\frac{\Lambda zt}{2b^2}})+\frac{\Lambda }{2b^2}\phi(z)\\&=\frac{\frac{5}{2b^2}}{z}\phi_1(z)+\frac{\Lambda }{2b^2}\phi(z)-\frac{z^{\frac{5}{2b^2}}}{z}\int_{\Gamma_\alpha}d[(t^2-1)^{\frac{1}{b^2}+1}]\exp{\frac{\Lambda zt}{2b^2}}\\
    &=\frac{\frac{5}{2b^2}}{z}\phi_1(z)+\frac{\Lambda }{2b^2}\phi(z)-\frac{z^{\frac{5}{2b^2}}}{z}(1+\frac{1}{b^2})\int_{\Gamma_\alpha}(t^2-1)^{\frac{1}{b^2}} 2t \exp{\frac{\Lambda zt}{2b^2}} dt\\
    &=\frac{\frac{5}{2b^2}}{z}\phi_1(z)+\frac{\Lambda }{2b^2}\phi(z)-\frac{2}{z}(1+\frac{1}{b^2})\phi_1(z).
\end{align*}
Combining the expressions for $\phi'(z)$ and $\phi'_1(z)$ then gives the desired result.
\end{proof}
From the form of the equation satisfied by $\varphi$ we immediately deduce the presence of a regular singularity at $z=0$ and, for non-zero $\Lambda$, an irregular singularity at infinity upon performing the change of variables $z \rightarrow \frac{1}{\tilde z}$.
\paragraph{Getting the Stokes matrix using properties of Bessel functions.}
Using the above linear equation, we can get a second order ODE that $\phi(z)$ satisfies,
\begin{align}
   4 b^4 z^2 \phi ''(z) + 4 b^2 z\left(2 b^2-3\right) \phi '(z)-\left(10 b^2+\Lambda^2 z^2-5\right)\phi (z) =0,
\end{align}
which can be solved in terms of Bessel functions (see Appendix \ref{sec:bessel}) as follows,
\begin{equation}
\varphi (z)= z^{\frac{-b^2+3}{2 b^2}}\begin{pmatrix}
 J_{\frac{b^2+2}{2 b^2}}\left(-\frac{i\Lambda  z}{2 b^2}\right) \\  J_{\frac{3b^2+2}{2b^2}}\left(-\frac{i\Lambda  z}{2 b^2}\right)
\end{pmatrix} \quad \textrm{and} \quad \varphi (z)= z^{\frac{-b^2+3}{2 b^2}}\begin{pmatrix}
 Y_{\frac{b^2+2}{2 b^2}}\left(-\frac{i \Lambda  z}{2 b^2}\right) \\  Y_{\frac{3b^2+2}{2b^2}}\left(-\frac{i \Lambda  z}{2 b^2}\right)
\end{pmatrix}, 
\end{equation}
or after a basis transformation,
\begin{align}\label{Besselsol}
    \varphi (z)= z^{\frac{-b^2+3}{2 b^2}}\begin{pmatrix}
 H_{\frac{3b^2+2}{2 b^2}}^{(1)}\left(-\frac{i \Lambda z}{2 b^2}\right) \\  \frac{d}{dz}H_{\frac{3b^2+2}{2b^2}}^{(1)}\left(-\frac{i \Lambda z}{2 b^2}\right) 
\end{pmatrix} \quad \text{and} \quad \varphi (z)= z^{\frac{-b^2+3}{2 b^2}}\begin{pmatrix}
 H_{\frac{3b^2+2}{2 b^2}}^{(2)}\left(-\frac{i\Lambda  z}{2 b^2}\right) \\  \frac{d}{dz}H_{\frac{3b^2+2}{2b^2}}^{(2)}\left(-\frac{i \Lambda z}{2 b^2}\right) 
\end{pmatrix}.
\end{align}
In the above we have used that $J_{\nu-1}(z) $ is a linear combination of $J_\nu(z)$ and $\frac{d}{dz}J_\nu (z)$ 
, and $H^{(1)}_\nu=J_\nu+i Y_\nu,H^{(2)}_\nu=J_\nu-i Y_\nu $ are Hankel functions. 

It is clear that there are two possible singularities of the functions, one is at zero, and the other at infinity. The first one being a regular singularity, and the latter being of irregular nature. It is possible to obtain the Stokes matrix at infinity using the relationship:
\begin{align}\label{Besselsol1}
   &H_\nu ^{(1)}(x e^{im\pi})=-\frac{\sin (\pi  (m-1) n)}{\sin (\pi  n)} H_\nu ^{(1)}(x)-\frac{e^{-i \pi  n} \sin (\pi  m n)}{\sin (\pi  n)} H_\nu ^{(2)}(x), \nonumber\\
   &H_\nu ^{(2)}(x e^{im\pi})=\frac{e^{i \pi  n} \sin (\pi  m n)}{\sin (\pi  n)} H_\nu ^{(1)}(x)+\frac{\sin (\pi  (m+1) n)}{\sin (\pi  n)} H_\nu ^{(2)}(x) .  
\end{align}
Together with the asymptotic behavior of the Hankel functions in some particular branches,
\begin{align}
    H_\nu ^{(1)}(x)=\hat{H}_\nu ^{(1)}(x)x^{-\frac{1}{2}} e^{ix}, H_\nu ^{(2)}(x)=\hat{H}_\nu ^{(1)}(x)x^{-\frac{1}{2}} e^{-ix},
\end{align}
where $\hat{H}_\nu ^{(1)}(x)$ has asymptotic series expansion in powers of $x^{-1}$, one deduces that there are two Stokes rays emanating from $x=0$ along the positive and negative real axis. That is, the asymptotic behavior changes when crossing these rays. From the expression of \eqref{Besselsol} we know that the Stokes matrices of these two solutions don't depend on the term $z^{\frac{-b^2+3}{2 b^2}}$. Thus the Stokes matrices are those from the Hankel functions alone.

From these expressions we can extract the Stokes matrix around $z=\infty$ when crossing the rays along the real axis \eqref{Besselsol1}. The result is the following expression encoding all Stokes matrices:
\begin{align}
S=  \begin{pmatrix}
1 & 2e^{i \pi  \frac{3b^2+2}{2 b^2}} \cos (\pi  \frac{3b^2+2}{2 b^2})\\ 
0 &1 
\end{pmatrix}.
\end{align}

\subsection{3-point function with one irregular operator}

When proving that our conformal block $\mathcal{F}$ satisfies the BPZ equation, we made crucial use of the identity \eqref{eq:phiddz}. We now will derive this identity (in a slightly modified form) and show that this equation is more naturally seen as part of KZ-type equations defining sections of irregular conformal block bundles.
In order to solve for the exact expression of $\mathcal{F}$, it suffices to obtain a differential equation for
\begin{align}
    \phi(z)&=\int_{\Gamma} dw \exp \left(\frac{\Lambda w+\log (w-z)+\log (w-1)+\log (w)}{b^2}\right)\nonumber\\&=\int_{\Gamma} dw e^{\frac{\Lambda w}{b^2}} ((w-1) w (w-z))^{\frac{1}{b^2}}\nonumber\\&=:\int_{\Gamma}dw \exp(I). 
\end{align}
\begin{prop} \label{thm:3pointKZ}
Let $\phi(z)$ be given as above and define
\begin{equation}
    \phi_1(z)\equiv \frac{1}{b^2}\int_{\Gamma}dw \exp(I) \frac{1}{w}, \quad \phi_2(z)\equiv\frac{1}{b^2}\int_{\Gamma}dw \exp(I) \frac{1}{w-1}.
\end{equation}
Then the vector $\varphi \equiv \begin{pmatrix}
\phi' \\ 
\phi_1\\
\phi_2  
\end{pmatrix}$ satisfies the following first order ODE,
\begin{equation}
    \frac{d \varphi}{dz} = A(z) \varphi,
\end{equation}
where the matrix $A(z)$ is given as follows,
\begin{align}
    A(z)=\left(
\begin{array}{ccc}
 \frac{ \Lambda}{b^2} & 0 & 0 \\
 0 & 0 & 0 \\
 0 & 0 & 0 \\
\end{array}
\right)+\frac{1}{z}\left(
\begin{array}{ccc}
 \frac{1}{b^2} & \frac{1}{b^2} & 0 \\
 \frac{1}{b^2} & \frac{1}{b^2} & 0 \\
 0 & 0 & 0 \\
\end{array}
\right)+\frac{1}{z-1}\left(
\begin{array}{ccc}
 \frac{1}{b^2} & 0 & \frac{1}{b^2} \\
 0 & 0 & 0 \\
 \frac{1}{b^2} & 0 & \frac{1}{b^2} \\
\end{array}
\right).
\end{align}
\end{prop}
\begin{proof}
We notice that,
\begin{align}
    \left\{\begin{matrix}
\phi '(z)=-\frac{1}{b^2}\int_{\Gamma}dw \exp(I) \frac{1}{w-z}\\ 
\phi ''(z)=\frac{1}{b^2}(\frac{1}{b^2}-1)\int_{\Gamma}dw \exp(I) \frac{1}{(w-z)^2}\\ 
\phi '''(z)=-\frac{1}{b^2}(\frac{1}{b^2}-1)(\frac{1}{b^2}-2)\int_{\Gamma}dw \exp(I) \frac{1}{(w-z)^3}
\end{matrix}\right. .
\end{align}
One straightforwardly computes 
\begin{align}
    \phi_1 '(z)&=-\frac{1}{b^4}\int_{\Gamma}dw \exp(I) \frac{1}{w}\frac{1}{w-z}\nonumber\\
               &=\frac{1}{z b^4}\int_{\Gamma}dw \exp(I)(\frac{1}{w}-\frac{1}{w-z})\nonumber\\
               &=\frac{1}{z}(\frac{1}{b^2}\phi_1(z)+\frac{1}{b^2}\phi'(z)). \label{eq:phi1prime}
\end{align}
Similarly, we have
\begin{align}
    \phi_2 '(z)=\frac{1}{z-1 }(\frac{1}{b^2}\phi_2(z)+\frac{1}{b^2}\phi'(z)). \label{eq:phi2prime}
\end{align}
Then, using integration by parts in the expression for $\phi''(z)$ , \eqref{eq:phi1prime}
and \eqref{eq:phi2prime}, one arrives at
\begin{align}
    \phi ''(z)&=\frac{1}{b^2}(\frac{1}{b^2}-1)\int_{\Gamma}dw \exp(I) \frac{1}{(w-z)^2}\nonumber\\
               &=\frac{1}{b^2}\int_{\Gamma} \exp(I') d((w-z)^{\frac{1}{b^2}-1})\nonumber\\
               &=-\frac{1}{b^2}\int_{\Gamma}(w-z)^{\frac{1}{b^2}-1} d(\exp (I'))\text{\space ,where $\exp(I^{'})=e^{\frac{\Lambda w}{b^2}}((w-1)w)^\frac{1}{b^2}$ }\nonumber\\
               &=-\frac{1}{b^2}\int_{\Gamma}dw\exp(I)(\frac{\Lambda}{b^2}+\frac{1}{b^2}\frac{1}{w}+\frac{1}{b^2}\frac{1}{w-1})\frac{1}{(w-z)}\nonumber\\
               &=\frac{\Lambda}{b^2}\phi'(z)-\frac{1}{b^4}\int_{\Gamma}dw\exp(I)(\frac{1}{w}+\frac{1}{w-1})\frac{1}{(w-z)}\nonumber\\
               &=\frac{\Lambda}{b^2}\phi'(z)+\phi_1'(z)+\phi_2'(z)\nonumber\\
               &=\frac{\Lambda}{b^2}\phi'(z)+\frac{1}{b^2(z-1)}\phi_2(z)+\frac{1}{b^2 z}\phi_1(z)+\frac{2 z-1}{b^2(z-1) z}\phi'(z),
\end{align}
which is
\begin{align}
    \phi''(z)=\frac{\Lambda}{b^2} \phi'(z)+\frac{1}{b^2(z-1)}\phi_2(z)+\frac{1}{b^2 z}\phi_1(z)+\frac{2 z-1}{b^2(z-1) z}\phi'(z). \label{eq:phi''}
\end{align}
Finally, by combining equations \eqref{eq:phi1prime}, \eqref{eq:phi2prime}, and \eqref{eq:phi''}, the claim follows.
\end{proof}

Another way to express the above, is the following third order ODE satisfied by $\phi_1$,
\begin{align*}
 &b^3 \left(b^2 (3 z-2)-\Lambda z^2+(\Lambda-4) z+2\right) \phi _1''(z)+\left(b^2-1\right) b \left(b^2-2 \Lambda z+\Lambda-3\right) \phi _1'(z)\\&+\frac{\left(b^2-1\right) \Lambda \phi _1(z)}{b}
 +b^5 (z-1) z \phi _1{}^{(3)}(z)=0,
\end{align*}
which can be easily derived using Theorem \ref{thm:3pointKZ}.
Returning to the matrix form of the equation, we see that the matrix $A(z)$ on the right has a very simple form with poles at $0$ and $1$ together with a constant term which gives rise to an irregular singularity at infinity. The matrix $A(z)$ can be viewed as the connection matrix of our conformal block bundle and in Section \ref{sec:generaldiff} we will show that it indeed satisfies the required flatness conditions. 

\paragraph{Asymptotic series expansion.} In the reminder of this section we will be concerned with solving the equation 
\begin{equation}
    \frac{d \varphi}{dz} = A(z) \varphi,
\end{equation}
in terms of an asymptotic series expansion in the limit $b \rightarrow 0$. To this end, we rewrite $A(z)$ as a series expansion $A(z)=\sum_{r=1}^\infty A_r z^{-r}$, where
\begin{align*}
    A_0=\left(
\begin{array}{ccc}
 \frac{ \Lambda}{b^2} & 0 & 0 \\
 0 & 0 & 0 \\
 0 & 0 & 0 \\
\end{array}
\right), \quad A_1=\left(
\begin{array}{ccc}
 \frac{2}{b^2} & \frac{1}{b^2} & \frac{1}{b^2} \\
 \frac{1}{b^2} & \frac{1}{b^2} & 0 \\
\frac{1}{b^2} & 0 & \frac{1}{b^2} \\
\end{array}
\right),\quad A_r=\left(
\begin{array}{ccc}
 \frac{1}{b^2} & 0 & \frac{1}{b^2} \\
 0 & 0 & 0 \\
 \frac{1}{b^2} & 0 & \frac{1}{b^2} \\
\end{array}
\right) \text{for  $r\geq 2$.}
\end{align*}
We can perform a basis transformation that takes our $A$ matrix to be a block diagonal matrix $B(z)$. Let us denote by $Y$ the column vector satisfying $\varphi=P(z)Y= P(z)\left(
\begin{array}{c}
 Y_1 \\
 Y_2 \\
 Y_3 \\
\end{array}
\right)$ for a suitable $P(z)$, such that $Y$ is subject to the following equation $$\frac{dY}{dz}=B(z)Y,$$ with $B$ a block diagonal matrix of the form $\left(
\begin{array}{ccc}
 B_1 & 0 & 0 \\
 0 & B_2 & B_3 \\
 0 & B_4 & B_5 \\
\end{array}
\right).$ Moreover, we require that $B(z)$ and $P(z)$ have series expansions \begin{align*}
    B_j(z)=\sum_{i=0}^\infty B_{ji} z^{-i}, \quad P(z)=\sum_{i=0}^\infty P_i z^{-i}.
\end{align*}
Now it is easy to see that $A$ and $B$ satisfy the following relationship
\begin{align}
    B=P^{-1}AP-P^{-1}\frac{dP}{dz}.
\end{align}
We can use this equation to solve for the matrices $B$ and $P$ order by order. We have
\begin{align}
    &P_1=\left(
\begin{array}{ccc}
 0 & -\frac{1}{\Lambda } & -\frac{1}{\Lambda } \\
 \frac{1}{\Lambda } & 0 & 0 \\
 \frac{1}{\Lambda } & 0 & 0 \\
\end{array}
\right), \quad P_2=\left(
\begin{array}{ccc}
 0 & \frac{b^2+1}{\Lambda ^2} & \frac{b^4+b^2-\Lambda }{b^2 \Lambda ^2} \\
 \frac{b^2-1}{\Lambda ^2} & 0 & 0 \\
 \frac{b^2 \Lambda +b^2-1}{\Lambda ^2} & 0 & 0 \\
\end{array}
\right),\nonumber\\
&P_3=\left(
\begin{array}{ccc}
 0 & \frac{-2 b^4-3 b^2+\Lambda +1}{\Lambda ^3} & \frac{2 b^2 \Lambda -2 b^6-3 b^4+b^2-\Lambda ^2+\Lambda }{b^2 \Lambda ^3} \\
 \frac{2 b^4-3 b^2-\Lambda -1}{\Lambda ^3} & 0 & 0 \\
 \frac{b^2 \Lambda ^2+2 b^4 \Lambda -b^2 \Lambda +2 b^4-3 b^2-1}{\Lambda ^3} & 0 & 0 \\
\end{array}
\right),
\end{align}
and \begin{align}
    B_0=\left(
\begin{array}{ccc}
 \frac{\Lambda}{b^2} & 0 & 0 \\
 0 & 0 & 0 \\
 0 & 0 & 0 \\
\end{array}
\right), \quad B_1=\left(
\begin{array}{ccc}
 \frac{2}{b^2} & 0 & 0 \\
 0 & \frac{1}{b^2} & 0 \\
 0 & 0 & \frac{1}{b^2} \\
\end{array}
\right), \quad B_2=\left(
\begin{array}{ccc}
 \frac{\Lambda +2}{b^2 \Lambda } & 0 & 0 \\
 0 & -\frac{1}{b^2 \Lambda } & -\frac{1}{b^2 \Lambda } \\
 0 & -\frac{1}{b^2 \Lambda } & \frac{\Lambda -1}{b^2 \Lambda } \\
\end{array}
\right),
\end{align}
and so on. Since $B_1$ is diagonal, we can solve for $Y_1$ by directly integrating as follows
\begin{align}
Y_1(z)=e^{\int_0^z B_1(t)dt}=e^{\frac{\Lambda}{b^2} z}z^{\frac{2}{b^2}}y_1(z),
\end{align}
where $y_(z)=1+\sum_{i=1}^\infty y_i z^{-i}\nonumber$ is determined by the matrix elements $B_{1j}$, giving
\begin{align}
    y_1(z)=\exp\left(\int_0^z \sum_{i=2}^\infty B_{1i}t^{-i} dt\right)=\exp\left( \sum_{i=2}^\infty \frac{1}{-i+1} B_{1i} z^{-i+1}\right)~.
\end{align}
We can get the series expansion of $y_1(z)$ from the above equation by substituting
\begin{align}
    &B_{12}=\frac{\Lambda+2}{b^2\Lambda},B_{13}=\frac{b^2 \Lambda ^2+b^4 \Lambda +2 b^4-2 b^2+\Lambda }{b^4 \Lambda ^2}\nonumber,\\
    &B_{14}=\frac{b^4 \Lambda ^2+b^2 \Lambda ^3+b^2 \Lambda ^2+2 b^6 \Lambda -b^4 \Lambda +4 b^6-6 b^4-2 b^2+\Lambda ^2-\Lambda }{b^4 \Lambda ^3},\nonumber\\
    &B_{15}=\frac{1}{b^4 \Lambda ^4}(-3 \Lambda-\Lambda ^2+\Lambda ^3+\text{higher order terms in $b$}),\nonumber\\
    &B_{16}=\frac{1}{b^6 \Lambda^5}(-\Lambda ^2+\text{higher order terms in $b$}).
\end{align}
In the following, we will occasionally also use the notation «$\cdots$» for higher order terms in $b$. Thus, further expanding on the exponent, we have
\begin{align}
    y_1(z)&=1+B_{12}\frac{1}{z}+\left(\frac{B_{13}}{-2}+\frac{1}{2!}B_{12}\right)\frac{1}{z^2}+\left(\frac{B_{14}}{-3}+\frac{1}{3!}B_{12}\right)\frac{1}{z^3}\nonumber\\
    &+\left(\frac{B_{15}}{-4}+\frac{1}{2!}\frac{B_{13}}{-2}+\frac{1}{4!}B_{12}\right)\frac{1}{z^4}+\left(\frac{B_{16}}{-5}+\frac{1}{5!}B_{12}\right)\frac{1}{z^5}+\text{higher order terms}
\end{align}
Analogously, we can solve for $Y_2$ and $Y_3$ using the same method, after employing the following transformation to diagonalize $B_2$,
\begin{align}
 \left(
\begin{array}{c}
 Y_2 \\
 Y_3 \\
\end{array}
\right)=\left(
\begin{array}{ccc}
   z^{\frac{1}{b^2}}\frac{1}{2} \left(\sqrt{\Lambda^2+4}+\Lambda\right)& z^{\frac{1}{b^2}}\frac{1}{2} \left(\Lambda-\sqrt{\Lambda^2+4}\right) \\
 z^{\frac{1}{b^2}} & z^{\frac{1}{b^2}} \\
\end{array}
\right)\left(
\begin{array}{c}
 Z_2 \\
 Z_3 \\
\end{array}
\right),
\end{align}
such that $Z$ satisfies the matrix equation
\begin{align}
    \frac{d}{dz}Z=C(z)Z, \quad C(z)=\sum_{i=2}^\infty C_i z^{-i},
\end{align}
with $C_2$ diagonal,
\begin{align}
    &C_2=\left(
\begin{array}{cc}
 -\frac{\sqrt{\Lambda^2+4}-\Lambda+2}{2 b^2 \Lambda} & 0 \\
 0 & \frac{\sqrt{\Lambda^2+4}+\Lambda-2}{2 b^2 \Lambda} \\
\end{array}
\right).
\end{align}
Higher order terms will have the form
\begin{align}
    C_3=\left(
\begin{array}{cc}
\frac{-\sqrt{\Lambda ^2+4}+\Lambda -2}{2 b^4 \Lambda  \sqrt{\Lambda ^2+4}}+... & \frac{-\sqrt{\Lambda ^2+4}+\Lambda -2}{2 b^4 \Lambda  \sqrt{\Lambda ^2+4}}+...\\
 \frac{\frac{\Lambda }{\sqrt{\Lambda ^2+4}}-\frac{2}{\sqrt{\Lambda ^2+4}}+1}{2 b^4 \Lambda }+... & -\frac{\sqrt{\Lambda ^2+4}+\Lambda -2}{2 b^4 \Lambda  \sqrt{\Lambda ^2+4}}+... \\
\end{array}
\right),~\cdots
\end{align}
With the above, we can also solve $Z_2$ and $Z_3$ recursively, in terms of regular Taylor series expansions in $\frac{1}{z}$. For example, assuming the following expansion of $Z_i$,
\begin{align}
Z_i=1+\sum_{j=1}^5 Z_{ij}z^{-j}+\text{higher order terms},
\end{align}
and substituting back into the original equation, we arrive at a recursive equation given order by order as
\begin{align}
\left\{\begin{matrix}
-jZ_{1j}=\sum_{i=0}^{j-1}(C_{j+1-i}^{1,1}Z_{1i}+C_{j+1-i}^{1,2}Z_{2i})\\
-jZ_{2j}=\sum_{i=0}^{j-1}(C_{j+1-i}^{2,1}Z_{1i}+C_{j+1-
i}^{2,2}Z_{2i})
\end{matrix}\right. ,
\end{align}
where $C_i^{j,k}$ denotes the $(j,k)$-th component of $C_i$. 
Thus we can solve for the  expression for $Z_{ij}$ in terms of the $C_i$ matrices. For example: $Z_{11}=C_2^{1,1}, Z_{21}=C_2^{2,2}$. For higher orders we have
\begin{align}
&Z_{12}=\frac{1}{4 b^4  \left(\Lambda ^2+4\right)}\Lambda ^2+...,\\
    &Z_{13}= -\frac{1}{12 b^6  \left(\Lambda ^2+4\right)}\Lambda ^2+...,\\
     &Z_{14}= -\frac{1}{48 b^8  \sqrt{\Lambda ^2+4}}\Lambda +...,\\
      &Z_{15}=\frac{1}{240 b^{10}  \sqrt{\Lambda ^2+4}}\Lambda +...,
\end{align}
and similarly
\begin{align}
  &Z_{22}=\frac{1}{4 b^4  \left(\Lambda ^2+4\right)}\Lambda ^2+...,\\
    &Z_{23}= \frac{1}{12 b^6  \left(\Lambda ^2+4\right)}\Lambda ^2+...,\\
     &Z_{24}= \frac{1}{48 b^8  \sqrt{\Lambda ^2+4}}\Lambda +...,\\
      &Z_{25}=-\frac{1}{240 b^{10}  \sqrt{\Lambda ^2+4}}\Lambda +...
\end{align}
Thus, from the above we can obtain the expressions for $Y_2$ and $Y_3$, namely
\begin{align}
    Y_2= z^\frac{1}{b^2}y_2(x), \quad Y_3= z^\frac{1}{b^2}y_3(x).
\end{align}
Here $y_i(x)$, $i=2,3$ have Taylor series expansions in terms of $\frac{1}{z}$: $y_i(x)=\sum_{j=0}^\infty  y_{ij}x^{-j}$ and is determined by $Z_i$ in the following way:
\begin{align}
    &y_{2i}=\frac{1}{2}(\sqrt{\Lambda^2+4}+\Lambda)Z_{1i}+\frac{1}{2}(-\sqrt{\Lambda^2+4}+\Lambda)Z_{2i}~,\\
     &y_{3i}=Z_{1i}+Z_{2i}~.
\end{align}
Summarizing, we finally arrive at the following form of the solution for the original equation
\begin{align}
    \begin{pmatrix}
\phi' \\ 
\phi_1\\
\phi_2  
\end{pmatrix}=P(z)\left(
\begin{array}{c}
 e^{\frac{\Lambda z}{b^2}}z^{\frac{2}{b^2}} y_1(z)  \\
 z^{\frac{1}{b^2}} y_2(z) \\
 z^{\frac{1}{b^2}} y_3(z) \\
\end{array}
\right).
\end{align}

\subsection{Flat connections from n-point functions}
\label{sec:generaldiff}

In the following, we generalize the above discussion to the case of arbitrary $n-$point functions corresponding to the Coulomb gas integrand \eqref{Correlationfunction}, using the same technique.

The superpotential we are interested in is of the following form:
\begin{align}
   \exp( \frac{1}{b^2}\mathcal{W}(w,z_1,z_2,...,z_n))=\prod_{i}((w-z_i)^{\frac{k_i}{b^2}})\prod_{i\neq j}(z_i-z_j)^{\frac{-k_i k_j}{2b^2}}\exp(\frac{\Lambda}{ b^2}(w-\frac{1}{2}\sum_i k_i z_i)).
\end{align}
The corresponding conformal blocks are then obtained by integration and can be written as below:
\begin{align*}
    \mathcal{F}&=\prod_{i}\exp(-\frac{\Lambda}{2 b^2}k_i z_i)\prod_{i\neq j}(z_i-z_j)^{\frac{-k_i k_j}{2b^2}}\int dw \prod_{i}((w-z_i)^{\frac{k_i}{b^2}})\exp(\frac{\Lambda}{ b^2}w)\\ &\equiv\prod_{i}\exp(-\frac{\Lambda}{2 b^2}k_i z_i)\prod_{i\neq j}(z_i-z_j)^{\frac{-k_i k_j}{2b^2}} \int dw \exp(I)\\&\equiv \prod_{i}\exp(-\frac{\Lambda}{2 b^2}z_i)\prod_{i\neq j}(z_i-z_j)^{\frac{-k_i k_j}{2b^2}} \phi(z),
\end{align*}
with $\phi(z)$ of the form \eqref{eq:phin-point}. We then have the following
\begin{theorem}
    Let $\Psi=\begin{pmatrix}
\frac{\partial \phi}{\partial z_1}\\ ...
\\\frac{\partial \phi }{\partial z_n}\\
\end{pmatrix}=\begin{pmatrix}
\Psi_1\\ ...
\\\Psi_n
\end{pmatrix}$.
Then the above implies the following equation satisfied by $\Psi$, namely
\begin{align} \label{eq:irrKZ}
\frac{\partial }{\partial z_i}\Psi=A_i \Psi +\sum_{j\neq i}\frac{\Omega_{ij}}{z_i-z_j}\Psi,
\end{align}
where we have defined 
\begin{align} \label{eq:A}
    A_i=
    \begin{blockarray}{cccccc}
 &  & i &  &   \\
\begin{block}{(ccccc)c}
 0 &...  &0  &...  & 0 &\\ 
... &...  &...  &...  &... &\\ 
0 &...  & \frac\Lambda {b^2} &...  & 0&i\\ 
... & ... &...  &...  &...& \\ 
0& ... & 0 &...  &0 &\\ 
\end{block}
\end{blockarray},
\end{align}
and
\begin{align} \label{eq:Omega}
   \Omega_{ij}=
    \begin{blockarray}{ccccccc}
 &  ...& j & ... & i  &... \\
\begin{block}{(cccccc)c}
 0 &0  &0  &...  & 0 &0&\\ 
0 &0  &0  &...  &0 &0&...\\ 
0 & ... &\frac{k_i}{b^2}  &...  &-\frac{k_j}{b^2}&0 &j \\ 
0 &...  & 0 &...  & ...&0 &...\\ 
0& ... & -\frac{k_i}{b^2} &...  &\frac{k_j}{b^2} &0 &i\\
0& 0 & ... &...  &0 &0 &...\\ 
\end{block}
\end{blockarray} 
\end{align}
Furthermore, the connection defined by the $\Omega_{ij}$ is flat.
\end{theorem}
\begin{proof}
We can get the differential equation satisfied by derivatives of $\phi$, by considering the following relation
\begin{align}
     &\frac{\partial \phi}{\partial z_i}=-\frac{k_i}{b^2}\int dw \exp(I)\frac{1}{w-z_i}, \\
     &\frac{\partial^2 \phi}{\partial z_i^2}=\frac{k_i}{b^2}\int  \exp(I^{'})d((w-z_i)^{\frac{k_i}{b^2}-1})\nonumber\\
     &=-\frac{k_i}{b^2}\int d(\exp(I^{'}))(w-z_i)^{\frac{k_i}{b^2}-1},\text{where exp$(I^{'})=\prod_{j\neq i}((w-z_j)^{\frac{k_j}{b^2}})\exp(\frac{\Lambda}{ b^2}w)$}\nonumber\\
     &=-\frac{k_i}{b^2}(-\frac{\Lambda}{ b^2}\frac{b^2}{k_i}\frac{\partial \phi}{\partial z_i}+\sum_{j\neq i}\frac{k_j}{b^2}(\int dw \exp(I)\frac{1}{w-z_i}\frac{1}{w-z_j}))\nonumber\\
     &=\frac{\Lambda}{ b^2}\frac{\partial \phi}{\partial z_i}-\frac{k_i k_j}{b^4}\sum_{j\neq i}\frac{1}{z_i-z_j}\int dw \exp(I)(\frac{1}{w-z_i}-\frac{1}{w-z_j})\nonumber\\
     &=\frac{\Lambda}{ b^2}\frac{\partial \phi}{\partial z_i}+\sum_{j\neq i}\frac{1}{z_i-z_j}(\frac{k_j }{b^2}\frac{\partial \phi}{\partial z_i}-\frac{ k_i}{b^2}\frac{\partial \phi}{\partial z_j}). \label{eq:phidoubleprime}
\end{align}
Similarly, we have by direct evaluation
\begin{align}
    \frac{\partial^2 \phi}{\partial z_i \partial z_j}=\frac{k_i k_j}{b^4}\int dw \exp(I)\frac{1}{w-z_i}\frac{1}{w-z_j}=\frac{1}{z_i-z_j}(\frac{k_j }{b^2}\frac{\partial \phi}{\partial z_i}-\frac{ k_i}{b^2}\frac{\partial \phi}{\partial z_j}).
\end{align}
We can then combine the above relations into one vector-valued equation proving the first part of the claim. Furthermore, we can show that in our case with $i,j\in\left\{ 1,2,3 \right \}$, the matrices $\Omega_{ij}$ will satisfy the flatness condition that should also be satisfied  in KZ equations and is an integrability condition for \eqref{eq:irrKZ}
 \begin{align}
     \Omega_{i,j}&=\Omega_{j,i},\\
 \left[\Omega_{i,j}+\Omega_{j,k},\Omega_{i,k}  \right ]&=0,\text{when $i,j,k$ are distinct.}\\
     \left [\Omega_{i,j},\Omega_{k,l}  \right ]&=0,\text{when $i,j,k,l$ are distinct.}
 \end{align}
For our particular case, we have
\begin{align}
    \Omega _{1,2}&= \Omega _{2,1}=\frac{1}{b^2}\left(
\begin{array}{ccc}
 k_2 & -k_1 & 0 \\
 -k_2 & k_1 & 0 \\
 0 & 0 & 0 \\
\end{array}
\right), \\
\Omega _{2,3}&=\Omega _{3,2}=\frac{1}{b^2}\left(
\begin{array}{ccc}
 0 & 0 & 0 \\
 0 & k_3 & -k_2 \\
 0 & -k_3 & k_2 \\
\end{array}
\right),\\
\Omega _{1,3}&=\Omega _{3,1}=\frac{1}{b^2}\left(
\begin{array}{ccc}
 k_3 & 0 & -k_1 \\
 0 & 0 & 0 \\
 -k_3 & 0 & k_1 \\
\end{array}
\right),
\end{align}
and one can easily deduce that
\begin{align}
    &\left [\Omega_{1,2}+\Omega_{2,3},\Omega_{1,3}  \right ]=0,\\
    &\left [\Omega_{1,3}+\Omega_{3,2},\Omega_{1,2}  \right ]=0,
\end{align}
meaning that the flatness condition also holds. This concludes the proof for the case of $3$-point functions. Using this result, it is then easy to inductively show the flatness condition for arbitrary $n$-point functions. 
\end{proof}
In this sense, we can see that this has a form resembling the KZ equation. This implies that the solutions of the equations \eqref{eq:irrKZ} are integrable, meaning that they give rise to $n$ different, globally independent solutions, which will correspond to our conformal blocks! 

Furthermore, we can check this result through counting the number of the Lefschetz thimbles. Here the number of Lefschetz thimbles , which corresponds to the number of BPS solitons of Landau-Ginzburg model, or the number of solutions of $\phi$, are given as the critical points of the superpotential $W$ with respect to $w$. And in our case, we have:
\begin{align}
    \frac{\partial W}{\partial w}=0\Leftrightarrow \sum_i \frac{k_i }{b^2}\frac{1}{w-z_i}+\frac{\Lambda}{b^2}=0.
\end{align}
This is a degree $n$ polynomial equation of $w$, which pocesses $n$ different solutions generically. That is in accordance with our KZ approach. 
\section{Conclusions}
\label{sec:conclusions}

In this paper we have studied Liouville conformal blocks with irregular vertex operators. Among the main results is the derivation of a flat connection characterizing sections of the corresponding irregular conformal block bundles. We find that these connections satisfy the flatness condition and are hence integrable. The corresponding ordinary differential equations resemble the KZ equations for WZW conformal blocks \cite{KNIZHNIK198483} up to a constant term which represents the irregular singularity at infinity. The corresponding quantum corrected level $\kappa=k+h^{\vee}$\footnote{$h^{\vee}$ is the dual Coxeter number of the symmetry group $G$ of the corresponding WZW model.} arising in KZ equations is replaced in our setup by the parameter $b^2$. Our asymptotic series expansions are therefore valid in the $b \rightarrow 0$ limit. Given the similarity to KZ equations, it is tempting to ask whether our irregular operators can be identified with irregular Kac-Moody representations in any way. We leave this task to future work. Moreover, it is desirable, given the explicit form of our connection, to derive monodromy representations for our conformal blocks. This can be achieved by finding all local monodromies, and transport matrices, as well as Stokes matrices at the irregular singularity which would then together fully fix the monodromy around singular points in moduli space \cite{JIMBO1981306}. From a more mathematical perspective, it would be interesting to connect to the works \cite{Feigin_2010a,Feigin_2010b,XuQG} where opers with irregular singularities and corresponding flat connections are studied. 

\section*{Acknowledgements}
We would like to thank Sergio Cecotti, Xia Gu, Sergei Gukov, Mauricio Romo, and Youran Sun  for valuable discussions. B.H. would also like to thank the Max-Planck Institute for Mathematics in Bonn, where part of this work was completed, for hospitality and financial support. B.H. and Y.L. were supported by the NSFC grant 12250610187, N.R. was supported by the Changjiang fund, and by BMSTC and ACZSP (Grant no. Z221100002722017).

\appendix

\section{Global Conformal Invariance and BPZ equation for general charges}
\label{sec:globalinv}
In this Appendix we show that the function $\mathcal{F}$ with expression: 
\begin{align}\label{eq:F3point}
    \mathcal{F}&=\Lambda^\beta \braket{V_{\frac{-1}{2b}}(z)V_{\frac{-k_0}{2b}}(z_0)V_{\frac{-k_1}{2b}}(z_1)I_{\alpha,\Lambda}(\infty)}\nonumber\\
    &=\Lambda^\beta(z-z_0)^{-\frac{k_0}{2 b^2}} (z-z_1)^{-\frac{k_1 k_0}{2 b^2}}(z_1-z_0)^{-\frac{1}{2 b^2}} e^{-\frac{\Lambda (z+k_1 z_1+k_0 z_0)}{2 b^2}}
    \nonumber\\
    &\times \int_{\Gamma} dw \exp \left(\frac{\Lambda w+\log (w-z)+k_1\log (w-z_1)+k_0\log (w-z_0)}{b^2}\right)\nonumber\\
    &\equiv f(z,z_0,z_1,\Lambda)\times \phi(z,z_0,z_1,\Lambda).
\end{align}
satisfies the primitive BPZ equation\eqref{eq:irrBPZ1}, the global conformal invariance equations for $L_{-1}$ \eqref{eq:L1inv} and $L_0$ \eqref{eq:L0inv}, thus the complete BPZ equation\eqref{eq:irrBPZ}. This will also allow us to fix the numerical constant $\beta$ in the overall $\Lambda^{\beta}$ factor. Notice that  we performed the rescaling \eqref{coordtrans}, 
so now  $c$ under the transformation should be  $\frac{\Lambda}{b}$.

Firstly, we  have the primitive BPZ equation:
  \begin{align}\label{eq:genirrBPZ1}
    \left((b^2\frac{\partial^2}{\partial z^2}+\frac{\Delta(k_1)}{(z_1-z)^2}+\frac{\Delta(k_0)}{(z_0-z)^2}-\frac{1}{z_1-z}\frac{\partial}{\partial z_1}-\frac{1}{z_0-z}\frac{\partial}{\partial z_0})-\frac{\Lambda^2}{4b^2}\right)f\phi=0,
\end{align}
where $\Delta(k_i)=-\frac{k_i}{2b}(Q+\frac{k_i}{2b})$. The proof is similar as that in \eqref{eq:irrBPZ2}. We just notice:
\begin{align}
     &\left((b^2\frac{\partial^2}{\partial z^2}+\frac{\Delta(k_1)}{(z_1-z)^2}+\frac{\Delta(k_0)}{(z_0-z)^2}-\frac{1}{z_1-z}\frac{\partial}{\partial z_1}-\frac{1}{z_0-z}\frac{\partial}{\partial z_0})-\frac{\Lambda^2}{4b^2}\right)f=0\\
     &  b^2\frac{\partial^2 \phi}{\partial z^2}=\Lambda \frac{\partial \phi}{\partial z}+\frac{1}{z-z_0}(k_0\frac{\partial \phi}{\partial z}-\frac{\partial \phi}{\partial z_0})+\frac{1}{z-z_1}(k_1\frac{\partial \phi}{\partial z}-\frac{\partial \phi}{\partial z_1})\\
      &\frac{\partial f}{\partial z}=-\frac{1}{2b^2}(\Lambda+\frac{k_0}{z-z_0}+\frac{k_1}{z-z_1})f
\end{align}
Now let us start with the derivation of $L_{-1}$ invariance, namely
\begin{align}\label{eq:genirrBPZ2}
    \left(\frac{\partial}{\partial z}+\frac{\partial}{\partial z_0}+\frac{\partial}{\partial z_1}+\frac{\Lambda}{b}(\alpha-Q) \right)\mathcal{F}=0.
\end{align}
To prove this, we follow our usual approach and apply the derivatives to our integral representation of $\mathcal{F}$. This gives
\begin{align}
    &\left(\frac{\partial}{\partial z}+\frac{\partial}{\partial z_0}+\frac{\partial}{\partial z_1} \right)\mathcal{F}= \left(\frac{\partial}{\partial z}+\frac{\partial}{\partial z_0}+\frac{\partial}{\partial z_1} \right)(f\phi) \nonumber\\
    &=\phi \left(\frac{\partial}{\partial z}+\frac{\partial}{\partial z_0}+\frac{\partial}{\partial z_1} \right)f+f\left(\frac{\partial}{\partial z}+\frac{\partial}{\partial z_0}+\frac{\partial}{\partial z_1} \right)\phi,
\end{align}
and we notice
\begin{align}
    &\frac{\partial f}{\partial z}=-\frac{1}{2b^2}(\Lambda+\frac{k_0}{z-z_0}+\frac{k_1}{z-z_1})f,\\
    &\frac{\partial f}{\partial z_1}=\frac{1}{2b^2}(\frac{k_1 f}{z-z_1}+\frac{k_1 k_0 f}{z_0-z_1})-\frac{k_1\Lambda}{2b^2}f,\\
    &\frac{\partial f}{\partial z_0}=\frac{1}{2b^2}(\frac{k_0 f}{z-z_0}+\frac{k_1 k_0 f}{z_1-z_0})-\frac{\Lambda}{2b^2}f,
\end{align}
thus giving
\begin{align}
    \frac{\partial f}{\partial z}+\frac{\partial f}{\partial z_1}+\frac{\partial f}{\partial z_0}=-\frac{(1+k_0+k_1)\Lambda}{2b^2} f.
\end{align}
After multiplication by $\phi$ this becomes
\begin{align}
  \frac{(1+k_0+k_1)\Lambda}{2b^2}f\phi +\phi(\frac{\partial f}{\partial z}+\frac{\partial f}{\partial z_1}+\frac{\partial f}{\partial z_0})=0.
\end{align}
On the other hand, using differentiation under the integral sign, we have the identity
\begin{align}
    &\left(\frac{\partial}{\partial z}+\frac{\partial}{\partial z_0}+\frac{\partial}{\partial z_1} \right)\phi= \left(\frac{\partial}{\partial z}+\frac{\partial}{\partial z_0}+\frac{\partial}{\partial z_1} \right)\int dw\exp(I)\nonumber\\
    &=\frac{1}{b^2}\int dw\exp(I)\left(-\frac{1}{w-z}-\frac{k_0}{w-z_0}-\frac{k_1}{w-z_1}\right)\nonumber\\
    &=\int dw\exp(I)\left(-\frac{\Lambda}{b^2}-\frac{1}{b^2}\left(-\frac{1}{w-z}-\frac{k_0}{w-z_0}-\frac{k_1}{w-z_1}\right)\right)+\frac{\Lambda}{b^2}\int dw\exp(I)\nonumber\\
    &=-\int d(\exp(I))+\frac{\Lambda}{b^2}\phi\nonumber\\
    &=\frac{\Lambda}{b^2}\phi,
\end{align}
which is just
\begin{align}
    \left(\frac{\partial}{\partial z}+\frac{\partial}{\partial z_0}+\frac{\partial}{\partial z_1} \right)\phi=\frac{\Lambda}{b^2}\phi.
\end{align}
Thus
\begin{align}
    \left(\frac{\partial}{\partial z}+\frac{\partial}{\partial z_0}+\frac{\partial}{\partial z_1}+\frac{\Lambda}{b}\left(\frac{1+k_0+k_1}{2b}-\frac{1}{b}\right) \right)\mathcal{F}=0.
\end{align}
 In order to check this relation, we first need to find the correct charge $\alpha$. By comparison we know that the constraint on $\alpha$ is
 \begin{align}
     \alpha-Q=\frac{1+k_0+k_1}{2b}-\frac{1}{b}\Leftrightarrow\alpha=b+\frac{1+k_0+k_1}{2b}.
 \end{align}
We now show that the $3$ point correlator will satisfy the $L_0$ invariance with the proper choice of $\beta$, namely
\begin{align}\label{eq:genirrBPZ3}
    \left(z\frac{\partial}{\partial z}+z_0\frac{\partial}{\partial z_0}+z_1\frac{\partial}{\partial z_1}+\Delta(1)+\Delta(k_1)+\Delta(k_0) -(\Lambda\frac{\partial}{\partial \Lambda}+\Delta_\alpha)\right)\mathcal{F}=0.
\end{align}
First of all, we just apply the operator $\left(z\frac{\partial}{\partial z}+z_0\frac{\partial}{\partial z_0}+z_1\frac{\partial}{\partial z_1}\right)$ to integral representation, giving
\begin{align}\label{L_0diff}
     \left(z\frac{\partial}{\partial z}+z_0\frac{\partial}{\partial z_0}+z_1\frac{\partial}{\partial z_1} \right)(f\phi)=\phi\left(z\frac{\partial}{\partial z}+z_0\frac{\partial}{\partial z_0}+z_1\frac{\partial}{\partial z_1} \right)f+f\left(z\frac{\partial}{\partial z}+z_0\frac{\partial}{\partial z_0}+z_1\frac{\partial}{\partial z_1} \right)\phi.
\end{align}
Notice that
\begin{align}
   (z\frac{\partial}{\partial z}+z_0\frac{\partial}{\partial z_0}+z_1\frac{\partial}{\partial z_1})f=(-\frac{k_0+k_1+k_0k_1}{2b^2}+\Lambda \frac{\partial}{\partial \Lambda}-\beta)f,
\end{align}
and also
\begin{align}
   &(z\frac{\partial}{\partial z}+z_0\frac{\partial}{\partial z_0}+z_1\frac{\partial}{\partial z_1})\phi=\frac{1}{b^2}\int dw \exp(I)(-\frac{k_0 z_0}{w-z_0}-\frac{k_1 z_1}{w-z_1}-\frac{z}{w-z})\nonumber\\
   &=\frac{1+k_0+k_1}{b^2}\phi-\frac{1}{b^2}\int dw \exp(I)w(\frac{k_0}{w-z_0}+\frac{k_1}{w-z_1}+\frac{1}{w-z})\nonumber\\
   &=\frac{1+k_0+k_1}{b^2}\phi+\int dw \exp(I) \frac{w\Lambda}{b^2}-\frac{1}{b^2}\int dw \exp(I)w(\Lambda+\frac{k_0}{w-z_0}+\frac{k_1}{w-z_1}+\frac{1}{w-z})\nonumber\\
   &=\frac{1+k_0+k_1}{b^2}\phi+\Lambda\frac{\partial }{\partial\Lambda}\int dw \exp (I)-\frac{\partial }{\partial\Lambda}\int dw \exp(I)(\Lambda+\frac{k_0}{w-z_0}+\frac{k_1}{w-z_1}+\frac{1}{w-z})\nonumber\\
   &=\frac{1+k_0+k_1}{b^2}\phi+\Lambda\frac{\partial }{\partial\Lambda}\int dw \exp (I)-b^2\frac{\partial }{\partial\Lambda}\int d(\exp(I))\nonumber\\
   &=\frac{1+k_0+k_1}{b^2}\phi+\Lambda\frac{\partial }{\partial\Lambda}\phi.
\end{align}
Combining these two equations finally gives
\begin{align}
    \left(z\frac{\partial}{\partial z}+z_0\frac{\partial}{\partial z_0}+z_1\frac{\partial}{\partial z_1} \right)(f\phi)&=\phi(-\frac{k_0+k_1+k_0k_1}{2b^2}+\Lambda \frac{\partial}{\partial \Lambda}-\beta)f +f(\frac{1+k_0+k_1}{b^2}\phi+\Lambda\frac{\partial }{\partial\Lambda})\phi \nonumber\\
     &=\Lambda\frac{\partial }{\partial\Lambda}(f\phi)-(\beta+\frac{k_0+k_1+k_0k_1}{2b^2}-\frac{1+k_0+k_1}{b^2})(f\phi).
\end{align}
And by comparison, we deduce that $\beta$ satisfies
\begin{align}
    &\beta+\frac{k_0+k_1+k_0k_1}{2b^2}-\frac{1+k_0+k_1}{b^2}=\Delta(1)+\Delta(k_1)+\Delta(k_0)-\Delta_\alpha \nonumber\\
    &\Rightarrow\beta=\Delta(1)+\Delta(k_1)+\Delta(k_0)+\frac{1+k_0+k_1}{b^2}-\alpha(b+\frac{1}{b}-\alpha)-\frac{k_0+k_1+k_0k_1}{2b^2}\nonumber\\
    &\Rightarrow\beta=-1.
\end{align}
So if we take $\beta=-1$ and $\alpha=b+\frac{1+k_0+k_1}{2b}$,  our corresponding conformal block with integral representation will satisfy the conformal identities \eqref{eq:genirrBPZ2} and \eqref{eq:genirrBPZ3}. Thus with \eqref{eq:genirrBPZ1}, \eqref{eq:genirrBPZ2} and \eqref{eq:genirrBPZ3}, following the same strategy as in the proof of \eqref{eq:irrBPZ}, the general conformal block $\mathcal{F}$ can be shown to satisfy the following BPZ equation after we substitute \eqref{eq:genirrBPZ2} and \eqref{eq:genirrBPZ3} into \eqref{eq:genirrBPZ1}, and set $z_1=1$, $z_0=0$:
\begin{align}\label{eq:genirrBPZ}
    &\left(b^2\frac{\partial^2}{\partial z^2}-(\frac{1}{z}+\frac{1}{z-1})\frac{\partial}{\partial z}+\frac{\Delta(k_1)}{(z-1)^2}+\frac{\Delta(k_0)}{z^2} \right. \nonumber\\
    &\left.+\frac{1}{z(z-1)}\left(c\frac{\partial}{\partial c}+\Delta_\alpha-\Delta(k_1)-\Delta(k_0)-\Delta(1)\right)-\frac{c(\alpha-Q)}{z}-\frac{c^2}{4}\right)\mathcal{F}=0.
\end{align}
Specifically, in the case of $k_0+k_1=-2b^2-1\Rightarrow\Delta_\alpha=0$, \eqref{eq:genirrBPZ} is exactly the same as the BPZ equation for irregular vectors appearing in \cite{Bonelli:2022ten}. And even in more general cases,  the BPZ equations for irregular vectors in  \cite{Bonelli:2022ten} can all be constructed from the integral form of the conformal blocks by modifying the value of $\beta$. 

\section{Bessel function and Stokes phenomenon from irregular singularity}
\label{sec:bessel}
Bessel functions $J_\alpha(z)$, $Y_\alpha(z)$ are linearly independent solutions of the following second order ordinary differential equation:
\begin{align}
z^2\frac{d^2f}{dz^2}+z\frac{d f}{dz}+(z^2-\alpha^2)f=0,
\end{align}
where $\alpha$ can be any complex number. The function $J_\alpha(z)$ is called the Bessel function of the first kind, and $Y_\alpha(z)$ the Bessel function of second kind. 

A particular linear combination of them is known as modified Bessel functions, and is defined as follows:
\begin{align}
    K_\alpha(z)=\frac{\pi}{2}\frac{i^{\alpha}J_{-\alpha}(iz)-i^{-\alpha}J_{\alpha}(iz)}{\sin(\alpha \pi)}.
\end{align}
Modified Bessel functions can be represented as an integral of $\exp(-2zt)(t^2-1)^{\alpha-\frac{1}{2}}z^\alpha $ with respect to $t$ along certain non-compact lines in the complex plane. We can see that the integrand resembles that of our 2-point conformal block.

A general second order  ordinary differential equation is of the form
\begin{align}
    \frac{d^2f}{d z^2} +p(z) \frac{df}{dz}+q(z)f=0,
\end{align}
with the coefficient functions $p(z)$ and $q(z)$ being Laurent polynomials on the complex plane. If either of these functions has a pole at a particular point $z=z_0$, we say that there is a \textit{singularity} at that point. One subdivides between \textit{regular singular} points, that is points where $(z-z_0)p(z)$ and $(z-z_0)^2q(z)$ are analytic around $z_0$, and \textit{irregular singular} points where this is not true. The presence of a singularity does not rule out the existence of solutions but indicates that solutions can have monodromies around singular points or, in the irregular case, only asymptotic expansions around those points can be found. In particular, near the irregular singularities, the solutions can be written as \cite{wasow1965asymptotic}:
\begin{align}
    f(z)=e^{Q(\frac{1}{z-z_0})}(\frac{1}{z-z_0)})^\rho\sum_{k=1}^\infty c_k (z-z_0)^k,
\end{align}
where $Q(z)$ is a polynomial function. As an example, the Bessel equation has a singularity at $z=0$ and a singularity at $z=\infty$, the first one being regular and the second one being irregular.

For regular singular points, solutions obtained as series expansions always converge in the neighborhood of $z_0$. And if we are interested in obtaining corresponding monodromies, we can easily read them off from the analytically continued expression for the solutions. But for irregular singularities, this approach fails, and we can only obtain solutions valid in certain regions, with different asymptotic behavior in different regions. These different regions are separated by co-dimension one lines in the $z$-plane known as \textit{Stokes lines} and upon crossing them, the corresponding bases of asymptotic solutions are connected by so called \textit{Stokes matrices}. This phenomenon is called the Stokes phenomenon.

\bibliographystyle{JHEP}
\bibliography{main}

\end{document}